\documentclass[a4paper,UKenglish,cleveref, autoref, thm-restate]{lipics-v2021}

\usepackage[utf8]{inputenc}
\usepackage[T1]{fontenc}
\usepackage{lmodern}
\usepackage{paralist}

\usepackage{amsmath} 
\usepackage{amssymb}
\usepackage{mathrsfs}
\usepackage{amsthm}
\usepackage{amsfonts}

\usepackage{graphicx}
\usepackage{algorithm}
\usepackage{algorithmic}
\usepackage{hyperref}
\usepackage{stmaryrd}
\usepackage{mathtools}
\usepackage[disable]{todonotes}
\hypersetup{hidelinks,breaklinks}
\usepackage{cleveref}
\usepackage{thm-restate}

\usepackage[normalem]{ulem}

\usepackage[notion, quotation, composition]{knowledge}

\input{pieces_of_knowledge.default.kl}

\knowledgenewrobustcmd{\cgtrans}[1]{\mathrel{\cmdkl{\xrightarrow{{\tiny #1}}}}}

\newcommand{\NN}{\mathbb{N}}

\newcommand{\cP}{\mathcal{P}}
\newcommand{\A}{\mathcal{A}}

\newcommand{\G}{\mathcal{G}}

\renewcommand{\L}{\mathcal{L}}

\newcommand{\Q}{\mathcal{Q}}

\renewcommand{\S}{\mathcal{S}}

\renewcommand{\leq}{\leqslant}
\renewcommand{\geq}{\geqslant}

\renewcommand{\implies}{~\Rightarrow~}

\newcommand{\trans}[1]{\xrightarrow{#1}}

\newrobustcmd\alphabet{\Sigma}

\hideLIPIcs
\nolinenumbers

\knowledgenewcommand{\concat}{\mathbin{\cmdkl{\cdot}}}

\knowledgenewcommand{\Lang}{\cmdkl{\L}}

\knowledgenewcommand{\cc}{\mathbin{\cmdkl{\cdot}}}

\knowledgenewrobustcmd\skel{\cmdkl{\mathrm{skel}}}

\knowledgenewcommand\LangS{\cmdkl{\L}}

\knowledgenewcommand{\Ai}{\cmdkl{A_D}}
\knowledgenewcommand{\De}{\cmdkl{\Delta}}

\knowledgenewrobustcmd\best[1][t]{\cmdkl{\mathrm{best}_{#1}}}
\knowledgenewrobustcmd\leader[1][t]{\cmdkl{\mathrm{leader}_{#1}}}

\knowledgenewrobustcmd\stateleq{\mathrel{\cmdkl{\sqsubseteq}}}
\newrobustcmd\stategeq{\mathrel{\kl[\stateleq]{\sqsupseteq}}}
\newrobustcmd\statel{\mathrel{\kl[\stateleq]{\sqsubset}}}
\newrobustcmd\stateg{\mathrel{\kl[\stateleq]{\sqsupset}}}

\knowledgenewrobustcmd\transleq{\mathrel{\cmdkl{\leqslant}}}
\knowledgenewrobustcmd\transgeq{\mathrel{\kl[\transleq]{\geqslant}}}
\knowledgenewrobustcmd\translt{\mathrel{\kl[\transleq]{<}}}

\knowledgenewcommand{\Si}{\cmdkl{\Sigma}}
\knowledgenewcommand{\ti}{\cmdkl{t_{\stateleq}}}
\knowledgenewcommand{\Sii}{\cmdkl{\Sigma_{\stateleq}}}
\knowledgenewcommand{\unit}[1][\Sigma]{\cmdkl{1_{#1}}}

\knowledgenewcommand{\ska}[1][t]{\cmdkl{\delta_{#1}}}

\newrobustcmd\dom{\mathrm{dom}}
\newrobustcmd\img{\mathrm{img}}
\knowledgenewrobustcmd\red{\cmdkl{\mathit{red}}}
\knowledgenewrobustcmd\green{\cmdkl{\mathit{green}}}

\knowledgenewcommand{\io}{\cmdkl{\iota}}

\knowledgenewrobustcmd\inter[1]{\cmdkl{\mathit{[#1]}}}

\knowledgenewrobustcmd\eps{\cmdkl{\varepsilon}}
\knowledgenewrobustcmd\Pieps{\cmdkl{\Pi_\varepsilon}}

\title{Eve-positional languages: putting order into Büchi automata}

\titlerunning{Eve-positional languages: putting order into Büchi automata} 

\author{Olivier Idir}{Université Paris Cité, CNRS, IRIF, France}{olivier.idir@ens-lyon.org}{https://orcid.org/0009-0003-3848-8515}{}

\authorrunning{O. Idir} 

\Copyright{Olivier Idir} 



\acknowledgements{I thank Thomas Colcombet for his many insights during the elaboration of this paper. I also thank Antonio Casares and Pierre Ohlmann for their enthusiastic answers concerning their article.}

\newcommand{\oii}[1]{\todo[color=green!30,inline]{#1}}

\begin{document}
	\maketitle

	\begin{abstract}
										An $\omega$-regular language is Eve-positional if, in all games with this language as objective, the existential player can play optimally without keeping any information from the previous moves. This notion plays a crucial role in verification, automata theory and synthesis.
										
										Casares and Ohlmann recently gave several characterisations of Eve-positionality of $\omega$-regular languages. For this, they introduce the notion of $\varepsilon$-complete parity automaton and show (among other results) that an $\omega$-regular language is Eve-positional if and only if it can be recognised by some $\varepsilon$-completion of a deterministic parity automaton.
										Colcombet and Idir built on their work, and obtained a more direct algebraic characterisation of Eve-positionality.
										
										We introduce a new formalism that characterises the Eve-positional languages, consisting of a restriction of non-deterministic Büchi automata. This allows us to complete a missing implication in Casares and Ohlmann's work. We then use this formalism to describe a determinization procedure for non-deterministic Büchi automata recognising such languages, with size blow-up at most factorial. We also show that this construction is state-wise optimal for languages over sufficiently complete alphabets.
	\end{abstract}

	\section{Introduction}

\subsection{Context}
The problem of reactive synthesis is the following: given a system equipped with a formal specification and interacting with its environment for a possibly infinite amount of time, we want to design a controller that can react to the environment while ensuring that its specification is met. 

A formalism frequently used for this is that of games on graphs. In these games, two players, Eve and Adam, move a token along the edges of an edge-coloured directed graph called an arena. Its vertices are partitioned between those belonging to Eve and those belonging to Adam, and when the token arrives in a vertex, its owner chooses the next edge, which the token will follow. This interaction continues infinitely, thus producing an infinite path through the arena and a corresponding infinite sequence of colours $w$ (also called an $\omega$-word). A play is won by Eve if the corresponding word belongs to some set $W$, called her objective. Else, it is won by Adam. Over the games we consider, it always stands that one of the players has a winning strategy, which describes a way to win the game even when this strategy is known in advance by the opponent.

The interaction between the system and its environment can be modeled by such a game, where the specification provides the winning objective $W$ \cite{Buchi_Landweber_synthesis}. In such a context, the complexity of the player's winning strategy is crucial, as games with simple winning strategies are generally easier to solve algorithmically and result in controllers that can be represented more succinctly.

\noindent \textbf{$\omega$-regular languages.}
An "$\omega$-regular language" is a set of "$\omega$-words" that can be described by many equivalent formalisms (e.g., automata that can be either deterministic, non-deterministic, or alternating; one-way or two-way; with Muller, Parity, Rabin, or Streett acceptance conditions; or by monadic-second-order logic). The class of "$\omega$-regular languages" plays a pivotal role in many decision procedures and mathematical arguments related to verification, automata theory, controller synthesis, etc.

\noindent \textbf{Eve-positionality.}
One of the simplest cases, when attempting to solve a game, is when the winning strategy is positional. This means that the corresponding player does not need any memory of the past to play optimally: they can make their decision simply by looking at the current state of the game. Note that this memoryless property need not be the same for both players. In this work, we are considering the objectives $W$ such that over all graphs, if Eve wins, she can do so with a positional strategy. They are called "Eve-positional languages".\\
Until recently, little was known about this class. Major progress was made in the last few years, but there was still lacking a convenient formalism to handle them, which we propose here.

\subsection{Contributions}
In this paper, we introduce the notion of "ordered Büchi automaton": these are non-deterministic Büchi automata such that Büchi-$\eps$-transitions form a total order on the states (these automata are very close cousins of the "$\eps$-complete" Büchi automata introduced by Casares and Ohlmann \cite{CO24Positional}).
\begin{enumerate}
	\item We establish that languages recognised by "ordered Büchi automata" are exactly the "Eve-positional" "$\omega$-regular languages" (\Cref{thm:oBuchi-iff-evePosi}),
	\item We propose a quadratic transformation from an "$\eps$-complete" "non-deterministic parity automaton" towards an "ordered Büchi automaton" recognising the same language (\Cref{thm:NPA-to-oB}).
	\item We provide a generic determinization procedure from "ordered Büchi automata" to deterministic parity automata, with at most factorial blow-up (\Cref{thm:exists-det-oB}). This determinization furthermore proves state-wise optimal over some languages (\Cref{lem:det-optim})
\end{enumerate}

Our first contribution implies taking a more detailed look at the "$\eps$-complete" "Büchi automata". They consist of parity automata with an underlying tree-like structure, such that they can be saturated with "$\eps$-transitions" following this structure without changing the recognised language. Intuitively, this structure allows us to keep track of how advantageous is the prefix read so far.
In order to gain a clearer view of the action of each letter, we quotient the states mutually reachable by non-Büchi $\eps$-transitions and follow the approach of transfer graphs (as in \cite{Bertrand_2019}). This allows us to view each letter as the set of transitions it induces in the automaton (called a "tile"), and observe that these "tiles" are upwards-closed for some order ${\transleq}$ over the transitions.
We obtain that the corresponding automata, called "ordered Büchi automata", recognise exactly the "Eve-positional languages". We believe that this formalism is useful both thanks to the ease with which it allows representing "Eve-positional languages", and as it provides a better structural understanding of these languages.

We then propose a transformation from any "$\eps$-complete" "non-deterministic parity automaton" with $n$ states and priorities $[0,2k-1]$ towards an "ordered Büchi automaton" recognising the same language, with $nk$ states.
According to Casares and Ohlmann, all "non-deterministic parity automata" recognising an "Eve-positional" language are "$\eps$-completable" (\Cref{theorem:co}); this transformation thus applies to all parity automata recognising such languages. 
This result allows, in conjunction with our first contribution, to complete a missing implication in their work: if a non-deterministic automaton is $\eps$-complete, then it recognises an "Eve-positional" language.
Note that in the general case, when going from "parity@parity automaton" to "Büchi automata", the known lower bound for this transformation is also of $\Theta(kn)$ \cite{Seidl_Niwinski_1999}; therefore, since being "$\eps$-completable" is purely a property of the recognised language, it is possible to go from an "$\eps$-completable automaton" to an "ordered Büchi automaton" without any additional state blow-up beyond what is already needed.

Our last contribution concerns the algorithmic investigation of "ordered Büchi automaton", and more precisely the question of determinizing such automata. We obtain that
a $n$-state "ordered Büchi automaton" can be effectively transformed into a "deterministic parity automaton" recognising the same language, with at most $1+\sum_{i=0}^{n-1} i!$ states.
We deduce from this that all "Eve-positional languages" can be determinized with at most factorial blow-up, compared to $O((1.64n)^n)$ in the general case \cite{determinising_parity}.
It is even possible to exploit some conditions on the alphabet in order to further reduce this blow-up. Notably, using a technique introduced in \cite{Safra_flowers}, we can establish the minimality of the resulting number of states over some languages, such as the Rabin language with $n$ pairs and all the possible letters with different behaviours (called the $n$-complete Rabin).

\subsection{Related works}

\noindent \textbf{Bipositionality.}
When looking at positionality, one can consider the bipositional languages, which are both Eve-positional and Adam-positional. We know from \cite{ColcombetN06} that the prefix-independent languages that are bipositional are exactly the parity languages, even when considering languages that are not necessarily "$\omega$-regular". There have been some attempts to generalize the latter result \cite{Bouyer_2023}, but a more general characterisation was only established in \cite[Theorem 7.1]{CO24Positional}.

\noindent \textbf{Eve-positionality.}
The "Eve-positionality" of some central classes was established as early as the 90s, such as parity languages \cite{Emerson_Jutla_parity} and Rabin languages \cite{RabinPosi}. However, the first thorough analysis of "Eve-positionality" was undertaken by Kopczi\'{n}ski \cite{Kopczynski06}.
He notably established that it is decidable whether a prefix-independent "$\omega$-regular language" is "Eve-positional"  \cite{Kopczynski06,Kopczynski07}, albeit with a quite slow procedure (in $O(n^{O(n^2)})$) that is not very informative on the reason why the language is "Eve-positional".
More generally, despite these breakthroughs, he did not obtain any general characterisation of "Eve-positional" languages.

Some partial characterisations were established in the following decade on specific subclasses of "$\omega$-regular languages". Colcombet, Fijalkow, and Horn established a semantic characterisation of "Eve-positionality" for safety languages \cite{CFK_safety_posi}, that is, languages where it is known in finite time whether some word is out of the language.
A characterisation for a larger class was actually obtained a few years later by Bouyer, Casares, Randour, and Vandenhove \cite{Bouyer_Buchi_posi}, as they obtained a semantic characterisation of "Eve-positional" languages recognised by deterministic Büchi automata. This, however, does not correspond to the general $\omega$-regular case, and for instance fails to cover the deterministic coBüchi case.

Recently, using the well-monotonic universal graphs initially developed by Ohlmann \cite{Ohlmann_universal_graphs}, he and Casares gave several other characterisations for "Eve-positionality" \cite{CO24Positional}. They also proposed new decision procedures, in polynomial time over the size of a deterministic automaton recognising the target language.
A key notion in their work is that of "$\eps$-completable" automata, which they introduced in the latter article.
Let us recall some of these characterisations here:
\begin{theorem}[3.1 and 3.2 in \cite{CO24Positional}]\label{theorem:co}
	Let $L \subseteq \alphabet^\omega$ be an "$\omega$-regular language". The following conditions are equivalent:
	\begin{asparaenum}
		\item $L$ is "Eve-positional" over all arenas (potentially infinite and containing $\eps$-moves).
		\item $L$ is "Eve-positional" over finite $\eps$-free Eve-only arenas.
		\item There exists a "$\eps$-completable" "deterministic parity automaton" recognising $L$.
		\item All (non-deterministic) "parity automata" for $L$ are "$\eps$-completable".
	\end{asparaenum}
\end{theorem}
These articles notably provide the first characterisations achieved over all "$\omega$-regular" languages. They however require one to reason with an automaton recognising the language, rather than directly with the language and inclusion properties.

Using their characterisation and notably the 1-to-2 player lift, Colcombet and Idir obtained an algebraic characterisation of these languages \cite{local-posi}, where a language is "Eve-positional" if and only if it satisfies three algebraic conditions (see \Cref{thm:local-posi}).
Each condition is of the form ``If certain word(s) belong to $L$, then at least one word obtained from their factors also belongs to $L$''. This supports the idea that "Eve-positionality" can be intuitively understood as establishing a total preference order among factors.
Furthermore, the first two of their conditions correspond exactly to the first two items of the characterisation obtained by Bouyer et al., and the third one is a generalisation of their third item that allows one to consider all "Eve-positional" "$\omega$-regular languages".

This last characterisation, however, 
does not provide a convenient formalism to handle "Eve-positional" languages.
Building on these last two results, we propose such a formalism, exactly capturing "Eve-positional" "$\omega$-regular languages" and that comes with a determinization procedure.

\subsection{Layout of the article}
We first recall a few standard definitions in \Cref{sec:defs}, along with the formalism of "tiles" to describe automata. We then define the desired "ordered Büchi automata" in \Cref{sec:Buchi tiles}. This section will also give a broad overview of this article's results, along with examples and a key intuition as to why these ordered Büchi automata characterise "Eve-positionality".

From \Cref{sec:eve-posi} onward, we will focus on the more technical aspects of our contribution. \Cref{sec:eve-posi} describes some semantic properties of these automata, and establishes the equivalence between "ordered Büchi automata" and "Eve-positional languages". It notably details the transformation from a "parity automaton" recognising an Eve-positional language towards an ordered Büchi automaton of same language (up to letter renaming).
Finally, \Cref{sec:det} describes the determinization procedure from an "ordered Büchi automaton" to a "deterministic parity automaton" recognising the same language, along with a size-optimality result.

	\section{General definitions}\label{sec:defs}

\subsection{Basic notations}

For a set A, its powerset is denoted $\cP(A)$.

\AP An ""alphabet"" is a finite set $\Pi$, whose elements are called letters.
The sets $\Pi^*$ and $\Pi^\omega$ respectively denote the sets of finite ""words"" and infinite \reintro*\kl{words} of length $\omega$ over $\Pi$ (also denoted ""$\omega$-words""). Subsets of $\Pi^*$ and $\Pi^\omega$ are called languages. An "$\omega$-word" is ""ultimately-periodic"" if it is of the form $u v^\omega$, with $u$ and $v$ finite words, and $v$ is non-empty.
The set of non-empty finite words is denoted $\Pi^+$, and the ""empty word"" is denoted $\intro*\io$. 
A language $L$ is ""prefix-independent"" if, for all $u\in \Pi^+$ and $w\in \Pi^\omega$, $w \in L$ if and only if $uw \in L$. 
A ""residual"" of a language $L$ is a language $u^{-1}L := \{w \mid u w \in L\}$, where $u$ is a finite word.

\AP The length of a finite "word" $u$ is written $|u|$. For a (possibly infinite) word $u$ and $n\in \NN$, $u_n$ denotes its letter at index $n$.

\AP An ""index"" $[i,j]$ is a non-empty finite interval of natural numbers $J = \{i, i+1,\dots, j\} \subseteq \NN$. For $n\in \NN$, we denote the index $[0,n-1]$ as $\intro*\inter{n}$.
Elements $c \in J$ are called ""priorities"". We say that an infinite sequence of "priorities" $(c_n)_{n\in \NN}$ is ""parity accepting"" (or simply \reintro*"accepting@@parity") if $\liminf_{n \to \infty} c_n \equiv 0 \mod 2$, else it is \reintro*"parity rejecting" (or \reintro*"rejecting@@parity").

For a function $f: A \to B$ and $b\in B$, we denote by $f^{-1}(b)$ the set $\{a \in A \mid f(a) = b\}$.

\subsection{Parity automata}

\AP A ""$J$-automaton"" (or \reintro*"automaton", or more generally a \reintro*"non-deterministic $J$-parity automaton"), with $J$ an "index", is a tuple $A = (\Pi, Q_A, I_A, \Delta_A)$, where $\Pi$ is an "alphabet", $Q_A$ the finite set of ""states"", $I_A \subseteq Q_A$ a set of ""initial states"", and $\Delta_A \subseteq Q_A \times \Pi \times J \times Q_A$ is the transition relation.
A ""transition"" $(q, a, c, q') \in \Delta_A$ is said to be from the "state" $q$ to the state $q'$, over the letter $a$ and with "priority" $c$. By default, all "automata" in consideration are complete\footnote{In the examples that follow, any missing transition is assumed to instead lead to a rejecting sink state.}, that is, for each state $q \in Q_A$ and letter $a\in \Pi$, there is at least one transition from $q$ over $a$ in $\Delta_A$.
We also consider automata to be restricted by default to the "states" accessible from the "initial states" by a sequence of "transitions".
The automaton $A$ is said ""deterministic"" if $I_A$ is a singleton and $\Delta_A$ can be described as a function from $Q_A \times \Pi$ to $J \times Q_A$.
When an automaton $A$ is known from the context, we skip the subscript and write just $Q,\Delta$, etc.

\knowledgerenewcommand\trans[1]{\cmdkl{\xrightarrow{#1}}}

\AP We denote by $\intro*\trans{a:c}$ the relation induced by "transitions" of parity $c$ over a letter $a$. That is, $p \trans{a:c} q$ denotes the existence of a "transition" $\delta = (p,a,c,q)\in\Delta$. We extend this notation to paths over words, where the label corresponds to the minimum priority encountered along this path.

For $w\in \Pi^*$ a word, the function $\Delta(q,w)$ is defined inductively over the length of $w$ as the set of states accessible from $q$ after a path over $w$, with $\Delta(q,\io):=\{q\}$.

\AP A ""run@@aut"" $\rho = (\delta_i)_{i\in \NN} \in \Delta^\omega$ over a word $w\in \Pi^\omega$ is a sequence of transitions $(q_i, w_i, c_i, q_{i+1})_{i\in \NN}$ such that $q_0 \in I$. A ""partial run"" is a finite prefix of a "run@@aut".
A "run@@aut" is called ""accepting@@run"" if its priority sequence $(c_i)_{i\in \NN}$ is "parity accepting" (recall that we consider the min-parity condition). A word $w\in \Pi^\omega$ is ""recognised@@aut"" by $A$ if there exists an "accepting run@@aut" over $w$ in $A$. The set of all "recognised words" of an automaton is called its ""language@@aut"", and is denoted $\intro*\Lang(A)$. Note that $\Lang(A)$ admits a finite number of different "residuals" (at most one for each state in $Q$).
A language $L \subseteq \Pi^\omega$ is called ""$\omega$-regular"" if there exists an "automaton" recognising it.

\AP We are particularly interested in the "$\inter{2}$-automata", which are also called ""Büchi automata"". It is well-known that all $\omega$-regular languages can be recognised by "Büchi automata" \cite{Rabin1968DecidabilityOS}. In such an automaton, a transition is called a ""Büchi transition"" if it has priority $0$.

\AP An ""$\eps$-automaton"" over an "alphabet" $\Pi$ (or "automaton with $\eps$-transitions") is an "automaton" over the alphabet $\intro*\Pieps := \Pi \uplus \{\eps\}$, where $\intro*\eps\notin \Pi$ is a distinguished letter. The ""language@@eps"" of an "$\eps$-automaton" $A$ is the set of words $w\in \Pi^\omega$ such that there exists $w' \in (\Pi \uplus \{\eps\})^\omega$ accepted by $A$ and such that $w$ is obtained from $w'$ by removing all occurrences of the letter $\eps$.

\subsubsection*{Tile automata}

\AP Another equivalent formalism for "automata" is to consider, for each letter, the set of transitions it entails. That is, for a "$J$-automaton" $A = (\Pi, Q, I, \Delta)$ and a letter $a\in \Pi$, we can consider $t_a := \{(p,c,q) \mid (p,a,c,q) \in \Delta\} \subseteq Q \times J \times Q$. We call this subset a ""tile"", and its elements are still called "transitions". For $t$ a "tile", a "transition" of $t$ is called ""horizontal"" if it is of the form $(q,c,q)$ for some $q\in Q$.

\AP Note that we can equip the set of all possible tiles $\Sigma := \cP(Q \times J \times Q)$ with a semigroup structure by setting the ""product"" $t_1 \intro*\cc t_2$ for "tiles" $t_1,t_2\in \Sigma$ as
\[
t_1 \cc t_2 := \{(p,\min(c_1,c_2),r)\mid \exists q,~( p,c_1,q ) \in t_1,~(q,c_2,r) \in t_2\}\ .
\]
This semigroup is actually a monoid, as $\intro*\unit[\Sigma] := \{ (q,\max(J),q) \mid q \in Q \}$ is a neutral element for $\cc$. Intuitively, the "product" allows us to consider directly, as a single tile, the behaviour of a non-empty word. This notably allows to observe more directly the $\omega$-semigroup for a given language.

\AP An infinite sequence of "tiles" $w = (w_j)_{j\in \NN}$ thus directly corresponds to a word $w_\Pi \in \Pi^\omega$ (or to many words, in the case where different letters map to the same tile, but they then all have the same behaviour in $A$), and a ""run@@tile"" over $w$ is an infinite sequence of transitions $(\delta_j = (q_j,r_j,q_{j+1}))_{j\in \NN}$ with $\delta_j \in w_j$ for all~$j$ and $q_0 \in I$.

\AP An "automaton" can be directly described by its alphabet $\Gamma$ of "tiles", instead of via its set $\Delta$ of "transitions". In that case, we can abstract away the corresponding letters, and only remember the induced tiles. We then call this automaton a ""tile automaton"", described as a tuple $\A = (Q,I,\Gamma)$, where $Q$ is a finite set of states, $I \subseteq Q$ is the set of "initial states", and $\Gamma \subseteq \Sigma$ is a set of "tiles"\footnote{We can omit the letter corresponding to the tile in this representation, as two letters with the same associated tile are functionally equivalent}. Its ""language@@system"" is still denoted $\intro*\LangS(\A)$. 

\begin{remark}\label{rem:langs-inclusion}
	Note that for $\A = (Q, I, \Gamma)$ an automaton, for $\A' := (Q,I', \Gamma')$ such that $I' \subseteq I$ and $\Gamma' \subseteq \Gamma$, we have that $\LangS(\A') \subseteq \LangS(\A)$. Indeed, any "accepting run@@tile" in $\A'$ still exists in $\A$.
\end{remark}

\subsection{Eve-positional languages}

In this article, we focus on a specific class of $\omega$-regular languages.

\AP A language $L$ is called ""Eve-positional"" if, whenever Eve wins a game with objective $L$, she can win with a positional strategy. In most of this article, however, we will not discuss games, and describe these languages via the characterisation established by Colcombet and Idir.

\phantomintro{local preference properties}
\begin{restatable}[\cite{local-posi}]{theorem}{EveposiLocal}\label{thm:local-posi} \AP
	An "$\omega$-regular language" ~$L\subseteq\Pi^\omega$ is "Eve-positional" if and only if, for all $u, u'$ finite words (possibly empty), $v,v'$ non-empty finite words, and $w,w'$ infinite words, the following "local preference properties" hold:
	\begin{asparaenum}
		\item if $uw\in L$ and $u'w'\in L$, then $uw' \in L$ or $u'w \in L$,
		\item if $uvw \in L$, then $uv^\omega \in L$ or $uw \in L$, and
		\item if $u(vv')^\omega \in L$, then $uv^\omega \in L$ or $uv'^\omega \in L$.
	\end{asparaenum}
\end{restatable}

When talking of "Eve-positional languages", a notion of particular relevance is the one of "$\eps$-complete automata", introduced by Casares and Ohlmann in \cite{CO24Positional}.

\begin{definition}[Casares, Ohlmann]\label{def:eps-complete}\AP
	A parity "$\eps$-automaton" $A = (\Pi, Q, I, \Delta)$ is ""$\eps$-complete"" if its "index" $J$ is of the form $\inter{2d+2}$, and
	\begin{asparaenum}
		\item the relations $\trans{\eps:1}, \trans{\eps:3}, \dots, \trans{\eps:d+1}$ define total preorders, each refining the previous one;
		\item for each even $c\in J$, the relation $\trans{\eps:c}$ is the strict variant of $\trans{\eps:c+1}$: for any $q,q'$, it holds that $q \trans{\eps:c} q'$ if and only if $q' \trans{\eps:c+1} q$ does not hold. 
	\end{asparaenum}
\end{definition}
Note that in the case of an "$\eps$-complete automaton", $I$ can be assumed downwards-closed for the total preorder $\trans{\eps:2d+1}$.

An automaton is ""$\eps$-completable"" if it is possible to add $\eps$-transitions to it to make it $\eps$-complete, without enlarging the "recognised language".

They notably established (\Cref{theorem:co}) that if an "automaton" $A$ recognises an "Eve-positional" language, it is "$\eps$-completable". However, they only proved the converse implication for deterministic automata: if some non-deterministic automaton $A$ is "$\eps$-completable", their result does not state that $\Lang(A)$ is "Eve-positional".

	\section{Ordered Büchi tiles and overview of the contributions}\label{sec:Buchi tiles}

In this section, we introduce the "ordered Büchi automata". They are "Büchi automata" with a total order ${\stateleq}$ over their set of states, such that at any point during a run it is always possible to decrease along ${\stateleq}$ whilst seeing an additional "Büchi transition". We then give a few examples, before giving an intuition of their link with "Eve-positional languages". We end this section with an overview of our contributions, and a discussion as to the interest of "ordered Büchi automata".

\subsection{Ordered Büchi languages}

\AP Let $Q$ be a finite set of states, and $\intro*\Si := \cP(Q \times \{0,1\} \times Q)$ its set of Büchi "tiles". We fix a total order $\intro*\stateleq$ over $Q$.

\AP Let $\delta=(p,c,q), \delta'= (p',c',q')$ be two transitions in $Q \times \{0,1\} \times Q$. We say that $\delta \intro*\transleq \delta'$ if $p \stateleq p'$, $q'\stateleq q$ (note the swap of order), and, if $(p,q) = (p',q')$, $c \leq c'$. We observe that ${\transleq}$ is a partial order. This captures the idea of downward Büchi "$\eps$-transitions" : intuitively $\delta \transleq \delta'$ if the transition $\delta'$ can be realized with $\eps$-transitions decreasing along ${\stateleq}$ and the transition $\delta$\footnote{The case where $c' = 1$ yet $c=0$ cannot be obtained in this manner, but is still included so that $\transleq$ is compatible with upward closure under $\cc$, as shown in \Cref{lem:Sii-monoid}.}

We define $\intro*\Sii$ as the set of upwards-closed tiles under the order ${\transleq}$, and the "tile" $\intro*\ti$ as the upward closure of $\unit$. 

\begin{restatable}{lemma}{lemSiimonoid}\label{lem:Sii-monoid}
	$(\Sii,\cc)$ is a sub-semigroup of $(\Si,\cc)$, and $(\Sii,\cc,\ti)$ is a monoid\footnote{It is not a submonoid of $(\Si,\cc,\unit)$ since the units do not coincide.}.
\end{restatable}

When considering an automaton $(Q,I, \Gamma)$ with $\Gamma \subseteq \Sii$, we may always assume that $I$ is downwards-closed, as the first transition taken can always be taken from some $q' \stateleq q$ for some $q\in I$.
More generally, due to this upward closure by ${\transleq}$, if a "partial run" over a "word" $w \in \Sii^+$ ends in some state $q$, then for all $q' \statel q$, there exists a run over $w$ ending in $q'$ that sees a "Büchi transition" (provided by upward closure in the last "tile").

\AP An ""ordered Büchi automaton"" consists of an automaton $\A := (Q, \stateleq, I, \Gamma)$ such that $(Q, I, \Gamma)$ is a "Büchi@@automaton" "tile automaton", and ${\stateleq}$ is a total order over $Q$ such that $I$ is downwards-closed for ${\stateleq}$ and $\Gamma \subseteq \Sii$.

\AP An "$\omega$-regular" language $L$ over some alphabet $\Pi$ is said to be an ""ordered Büchi language"" if it is the "language@@system" of some ordered Büchi automaton, up to renaming the letters in $\Pi$. That is, there exists an "ordered Büchi automaton" $\A = (Q,\stateleq, I, \Gamma)$ and a morphism $f:\Pi \to \Gamma$ (extended to words) such that for all $w\in \Pi^\omega$, $w\in L$ if and only if $f(w)\in \LangS(\A)$.

We finally introduce the notion of "skeleton" of a "tile", that allows for a more succinct description. 
We can define a unique minimal generator for each "tile" of~$\Sii$, which we shall call "skeleton".

\begin{restatable}[existence of the skeleton]{lemma}{lemskel}\label{lemma:skeleton}
	For all tiles $t$ in $\Sii$, there exists a unique minimal (for inclusion) tile $s$ such that $t$ is the upward closure of $s$.
	This tile $s$ furthermore satisfies that for all distinct "transitions" $p\trans{s}p'$ and $q\trans{s} q'$, both $p\neq q$ and $p'\neq q'$.
\end{restatable}

\AP Let $t$ be a tile in $\Sii$, the minimal "tile" $s$ as described in \Cref{lemma:skeleton} is called the ""skeleton"" of~$t$ and is denoted $\intro*\skel(t)$.

\subsection{Examples}

We now provide a few examples of such "ordered Büchi automata", in order to clarify the objects and how they are handled.

\begin{figure}[!htb]
	\begin{minipage}[t]{0.35\textwidth}
		\includegraphics[height=2.8cm]{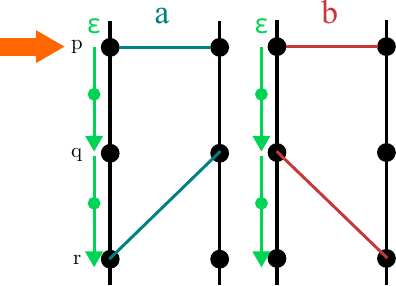}
		\caption{\label{fig:oB-inf-aa-fin-bb} An ordered Büchi automaton recognising the words with infinitely often a factor $aa$ and only finitely often $bb$.}
	\end{minipage}
	\hfill
	\begin{minipage}[t]{0.60\textwidth}
		\centering
		\includegraphics[height=2.8cm]{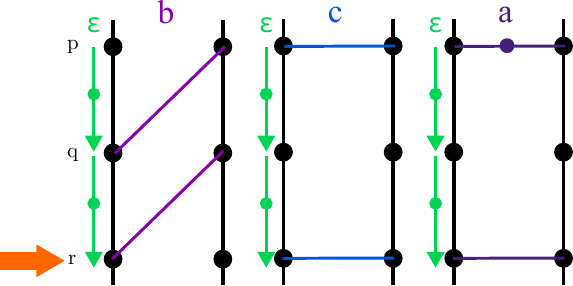}
		\caption{\label{fig:inf-b-ou-bb-inf-a} An ordered Büchi automaton recognising the words with infinitely many $b$'s, or a factor $bb$ and infinitely many $a$'s.}
	\end{minipage}
	\caption*{In these figures, "ordered Büchi automata" are represented by the "skeletons" of the letters in their corresponding alphabet, with the order ${\stateleq}$ represented vertically (the highest element being on top). The highest element of the starting set $I$ is designated by the orange arrow, and Büchi transitions are represented as lines with a circle in the middle. The $\eps$-labelled arrows on the left are there to recall that it is always possible, using the upward closure for $\transleq$, to pass to a smaller vertex while seeing a Büchi transition. As this move can take place between any transitions, they are thus equivalent to $\eps$-transitions.}
\end{figure}

The language recognised by the "ordered Büchi automaton" in \Cref{fig:oB-inf-aa-fin-bb} is the set of words with only a finite number of times the factor $bb$, 
and infinitely often the factor $aa$. Indeed, an accepting run in this automaton cannot remain infinitely often in the state $p$, and thus at some point goes in $q$ or $r$. 
Past this point, it can no longer see a factor $bb$, which would end the run, and can only see a Büchi transition when seeing an $a$ whilst the run is in $q$ -- hence an infinity of factor $aa$.\\
We similarly observe that the "ordered Büchi automaton" described in \Cref{fig:inf-b-ou-bb-inf-a} is accepting in the case of a word with infinitely many $b$'s, as it can then always see a Büchi transition coming back to $r$, or in the case of some factor $bb$, allowing it to reach the state $p$, followed by a word with an infinity of $a$'s.

\AP \noindent\textbf{Rabin languages are ordered Büchi.} We also observe that "ordered Büchi automata" can recognise "Rabin languages". A language $L$ over an alphabet $\Pi$ is a ""Rabin language"" if there exists a finite set of pairs $(G_i,R_i)_{i<n}$, with $G_i, R_i \subseteq \Pi$, such that a word $w$ belongs to $L$ if and only if, for $\inf(w)$ the set of its infinitely recurring letters, there exists some $i<n$ such that $\inf(w) \cap G_i \neq \emptyset$ and $\inf(w)\cap R_i = \emptyset$.

Let $L$ be such a "Rabin language" over an alphabet $\Pi$, with Rabin pairs $(G_i,R_i)_{i<n}$.
We define $\A = (\inter{n+1}, \leq, \inter{n+1}, \Gamma)$ an "ordered Büchi automaton", along with a morphism $f : \Pi \to \Gamma$, and will prove that it recognises exactly $L$, up to renaming the letters with $f$.
For $a$ a letter in $\Pi$, we define the following "tile":
\begin{align*}
	t_a:=  \{(n,1,n)\}\uplus \{(i,0,i) \mid i < n, a \in G_i\setminus R_i \} \uplus \{(i,1,i) \mid i<n, a \notin R_i\}\ .
\end{align*}
Note that $t_a$ is not necessarily a "skeleton", as it can have two transitions starting from the same state $i$.
We then define $f(a)$ as the upwards closure of $t_a$, and $\Gamma$ as $f(\Pi)$.
An example of such a construction can be observed in \Cref{fig:rabin-oBuchi}.

\begin{figure}[!htb]
	\centering
	\includegraphics[height=2.8cm]{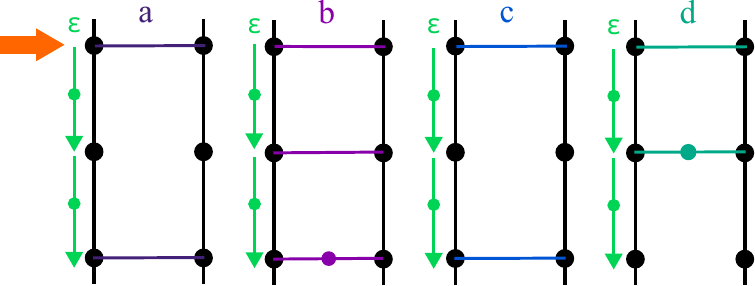}
	\caption{\label{fig:rabin-oBuchi} The following tiles describe the Rabin language of pairs $(\{d\}, \{a,c\})$ and $(\{b\},\{d\})$.}
\end{figure}

We claim that for $w = (w_j)_{j\in \NN}\in \Pi^\omega$, $w \in L$ if and only if $f(w)\in \LangS(\A)$.
Indeed, if $w \in L$, there exists $i$ such that infinitely often, $w_j \in G_i$, and past some index $j_0$, $w_j \notin R_i$. One easily checks that  $(n,1,n)^{j_0-1}\ (n,0,i)\ (i,r_{j_0},i)\ (i,r_{j_0+1},i)\dots$
(in which for all~$j\geqslant j_0$, $r_j$ is~$0$ if $w_j\in G_i$, and $1$ otherwise), is an "accepting run@@tile" over $w$.

Conversely, suppose that $f(w)\in \LangS(\A)$, with corresponding run $((q_j,r_j,q_{j+1}))_{j \in \NN}$. Note that by construction of $f$, for all $a \in \Pi$, if $i \trans{f(a)} i'$, then $i' \leq i$. Therefore, the sequence $(q_j)_{j\in \NN}$ eventually stabilizes in some $i_\infty$ at some index $j_0$ (with $i_\infty < n$, as there is no "horizontal Büchi transition" over the state $n$). Thus, for all $j\geq j_0$, $w_j \notin R_{i_\infty}$, and as the run is accepting, infinitely often, $(i_\infty,0,i_\infty) \in f(w_j)$, thus infinitely often $w_j \in G_{i_\infty}$. Therefore $i_\infty$ is an accepting Rabin pair, and $w\in L$.

\subsection{Ordered Büchi languages are the Eve-positional languages}

The main interest of these "ordered Büchi automata" lies in their ability to exactly recognise "Eve-positional" "$\omega$-regular languages":

\begin{restatable*}{theorem}{oBiffEveposi}\label{thm:oBuchi-iff-evePosi}
	Let $L$ be an $\omega$-regular language. It is "Eve-positional" if and only if it is "ordered Büchi".
\end{restatable*}

This theorem will be proved in \Cref{sec:eve-posi}, but we already provide the intuition of why, for $\A = (Q, \stateleq, I, \Gamma)$ an "ordered Büchi automaton" recognising a "prefix-independent" language $L := \LangS(\A)$, this language is "Eve-positional".

It thus suffices to establish that $L$ satisfies the "local preference properties" (stated in \Cref{thm:local-posi}). By "prefix-independence", we only need to prove that it satisfies the third one. That is, for all $v,v' \in \Gamma^+$, if $(vv')^\omega \in L$, either $v^\omega \in L$ or $v'^\omega \in L$.

Without loss of generality, we can assume that all the states of $Q$ are reachable at any point along a run, as $L$ is "prefix-independent". We first claim that for $u\in \Pi^+$, $u^\omega \in L$ if and only if there exists a "horizontal Büchi transition" along $u$. The converse implication is immediate. For the direct sense, we observe that $\ti^\omega\notin \LangS(\A)$, as there is no infinite path in $\ti^\omega$ witnessing an infinity of "Büchi transitions". Therefore, up to considering $t_u$ as the tile obtained by "product" of the tiles of $u$, necessarily $t_u \nsubseteq \ti$. We then observe that any transition in $t_u \setminus \ti$ induces a "horizontal Büchi transition" by upward closure, which establishes the desired claim.

In consequence, for $v,v'\in \Sii^+$ such that $(vv')^\omega \in L$, we are in one of the four cases described in \Cref{fig:Eve-posi-PI}. In the first and second cases, immediately, either $v^\omega \in L$ or $v'^\omega\in L$. In the third case, by upward closure, we observe that $q\trans{v:0}q$, and similarly in the fourth case $q \trans{v':0} q$, which concludes the proof.

\begin{figure}[!htb]
	\centering
	\includegraphics[width=0.95\textwidth]{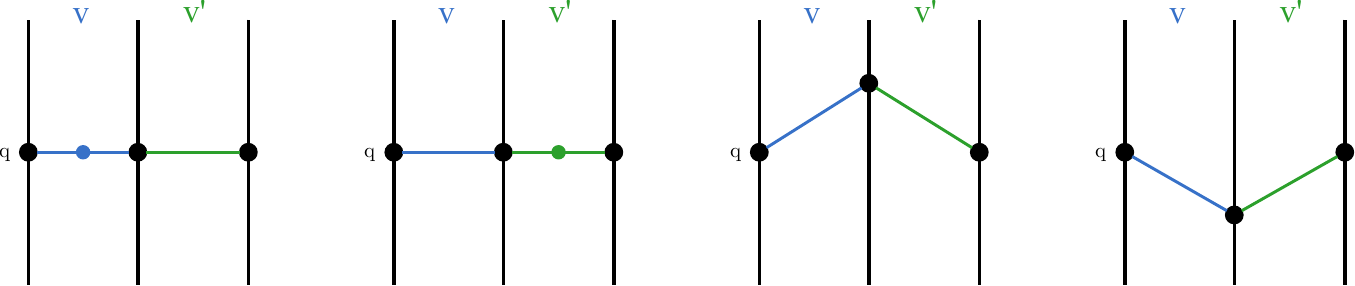}
	\caption{\label{fig:Eve-posi-PI} If $(vv')^\omega \in L$, necessarily one of the following holds for some state $q\in Q$. In the third and fourth cases, one of the two transitions depicted (over $v$ or $v'$) is necessarily Büchi, but we do not use this property for the proof.}
\end{figure}

We can prove that this result also holds in the non-"prefix-independent" case, using again the same kind of tile manipulations.

\begin{restatable*}{proposition}{oBEveposi}\label{lem:oBuchi-are-evePosi}
	Let $\A = (Q, \stateleq, I, \Gamma)$ be an "ordered Büchi automaton". Then $\LangS(\A)$ is "Eve-positional".
\end{restatable*}

The converse inclusion is shown by exhibiting a transformation from any "$\eps$-complete" "Büchi automaton" to an "ordered Büchi automaton" recognising the same language (up to letter renaming): as all "Eve-positional" "$\omega$-regular languages" can be recognised by an "$\eps$-complete" "Büchi automaton" \cite{CO24Positional}, this suffices.

\begin{restatable*}{proposition}{BuchitooB}\label{lem:Buchi-to-oB}
	For $L$ a language recognised by an "$\eps$-completable" "Büchi automaton", $L$ is "ordered Büchi".
\end{restatable*}

We even provide a transformation from a general "$\eps$-complete" "parity automaton" to an "ordered Büchi automaton" recognising the same language, that is asymptotically optimal.

\begin{restatable*}{theorem}{NPAtooB}\label{thm:NPA-to-oB}
	For $A$ an "$\eps$-completable" "$\inter{2k}$-automaton" with $n$ states, recognising a language $L$, there effectively exists an "ordered Büchi automaton" with $kn$ states recognising $L$.
\end{restatable*}

A corollary of these results is that we can complete the missing implication in Casares and Ohlmann's work. Indeed, chaining \Cref{thm:NPA-to-oB} and \Cref{lem:oBuchi-are-evePosi}, we obtain that:
\begin{restatable*}{corollary}{epscompl}\label{lem:eps-complete-CO}
	Let $A$ be an "$\eps$-completable automaton", then $\Lang(A)$ is "Eve-positional".
\end{restatable*}

\subsection{Determinization of ordered Büchi automata}

We finally provide a determinization procedure for "ordered Büchi automata", with size blow-up at most factorial. This determinization implies maintaining a kind of Last Appearance Record (initially introduced in \cite{LAR}), where we keep track, after some prefix, of the best run candidate ending in each state, ordered by comparative longevity. The states of this automaton, denoted "records", thus correspond to permutations of downwards-closed subsets of $Q$.  The \Cref{sec:det} is dedicated to this determinization and the corresponding subresults.

\begin{restatable*}{theorem}{detoB}\label{thm:exists-det-oB}
	For $\S = (Q, \stateleq, I, \Gamma)$ an "ordered Büchi automaton" of language $L$, there effectively exists a "deterministic parity automaton" recognising $L$, with at most $$1 + \sum_{i=0}^{|Q|-1} i!$$ states.
\end{restatable*}

This result notably implies that all Büchi automata recognising "Eve-positional languages" can be determinized in size at most factorial, as opposed to the lower bound in $O((1.64n)^n)$ in the general case \cite{determinising_parity}.

We finally show that over some sufficiently complete tile alphabets, such as for the $n$-complete Rabin (which is the Rabin language with $n$ pairs and all the letters with different behaviours), this determinization procedure yields automata optimal as to their number of states.

\subsection{Discussion}
We introduce in this article "ordered Büchi automata" as an alternative formalism to characterise "Eve-positional languages". They allow a succinct yet easy-to-handle representation of these languages.
Using this formalism, we propose a determinization procedure, optimal over some languages (as soon as the tile alphabet contains an expressive enough sub-alphabet), which establishes a new upper bound on the determinization of "Eve-positional languages". In particular, this determinization makes explicit a structural simplicity of these languages, as it only requires to store permutations of the states instead of the Safra trees used in the general case (\cite{safra_determinization_88}).

A continuation of this work that looks promising is to study the same idea over larger classes than "$\omega$-regular languages", as in \cite{mem_CO_2025}. This would amount to considering an "ordered Büchi automaton" with $\omega$ states.\\
Unfortunately, we did not manage, using these "ordered Büchi automata", to obtain a faster decision procedure than the one proposed by Casares and Ohlmann \cite{CO24Positional}.
	
	\bibliography{biblio}

\begin{thebibliography}{10}

\bibitem{Bertrand_2019}
Nathalie Bertrand, Miheer Dewaskar, Blaise Genest, Hugo Gimbert, and Adwait~Amit Godbole.
\newblock Controlling a population.
\newblock {\em Logical Methods in Computer Science}, Volume 15, Issue 3, jul 2019.
\newblock URL: \url{http://dx.doi.org/10.23638/LMCS-15(3:6)2019}, \href {https://doi.org/10.23638/lmcs-15(3:6)2019} {\path{doi:10.23638/lmcs-15(3:6)2019}}.

\bibitem{Bouyer_Buchi_posi}
Patricia Bouyer, Antonio Casares, Mickael Randour, and Pierre Vandenhove.
\newblock {Half-Positional Objectives Recognized by Deterministic B\"{u}chi Automata}.
\newblock In Bartek Klin, S{\l}awomir Lasota, and Anca Muscholl, editors, {\em 33rd International Conference on Concurrency Theory (CONCUR 2022)}, volume 243 of {\em Leibniz International Proceedings in Informatics (LIPIcs)}, pages 20:1--20:18, Dagstuhl, Germany, 2022. Schloss Dagstuhl -- Leibniz-Zentrum f{\"u}r Informatik.
\newblock URL: \url{https://drops.dagstuhl.de/entities/document/10.4230/LIPIcs.CONCUR.2022.20}, \href {https://doi.org/10.4230/LIPIcs.CONCUR.2022.20} {\path{doi:10.4230/LIPIcs.CONCUR.2022.20}}.

\bibitem{Bouyer_2023}
Patricia Bouyer, Mickael Randour, and Pierre Vandenhove.
\newblock Characterizing omega-regularity through finite-memory determinacy of games on infinite graphs.
\newblock {\em TheoretiCS}, Volume 2, jan 2023.
\newblock URL: \url{http://dx.doi.org/10.46298/theoretics.23.1}, \href {https://doi.org/10.46298/theoretics.23.1} {\path{doi:10.46298/theoretics.23.1}}.

\bibitem{Buchi_Landweber_synthesis}
J.~Richard Büchi and Lawrence~H. Landweber.
\newblock Solving sequential conditions by finite-state strategies.
\newblock {\em Transactions of the American Mathematical Society}, 138:295--311, 1969.
\newblock URL: \url{http://www.jstor.org/stable/1994916}.

\bibitem{CO24Positional}
Antonio Casares and Pierre Ohlmann.
\newblock Positional {\(\omega\)}-regular languages.
\newblock In Pawel Sobocinski, Ugo~Dal Lago, and Javier Esparza, editors, {\em Proceedings of the 39th Annual {ACM/IEEE} Symposium on Logic in Computer Science, {LICS} 2024, Tallinn, Estonia, July 8-11, 2024}, pages 21:1--21:14. {ACM}, 2024.
\newblock \href {https://doi.org/10.1145/3661814.3662087} {\path{doi:10.1145/3661814.3662087}}.

\bibitem{mem_CO_2025}
Antonio Casares and Pierre Ohlmann.
\newblock Characterising memory in infinite games.
\newblock {\em Logical Methods in Computer Science}, Volume 21, Issue 1, March 2025.
\newblock URL: \url{http://dx.doi.org/10.46298/lmcs-21(1:28)2025}, \href {https://doi.org/10.46298/lmcs-21(1:28)2025} {\path{doi:10.46298/lmcs-21(1:28)2025}}.

\bibitem{CFK_safety_posi}
Thomas Colcombet, Nathanaël Fijalkow, and Florian Horn.
\newblock {Playing Safe}.
\newblock In Venkatesh Raman and S.~P. Suresh, editors, {\em 34th International Conference on Foundation of Software Technology and Theoretical Computer Science (FSTTCS 2014)}, volume~29 of {\em Leibniz International Proceedings in Informatics (LIPIcs)}, pages 379--390, Dagstuhl, Germany, 2014. Schloss Dagstuhl -- Leibniz-Zentrum f{\"u}r Informatik.
\newblock URL: \url{https://drops.dagstuhl.de/entities/document/10.4230/LIPIcs.FSTTCS.2014.379}, \href {https://doi.org/10.4230/LIPIcs.FSTTCS.2014.379} {\path{doi:10.4230/LIPIcs.FSTTCS.2014.379}}.

\bibitem{local-posi}
Thomas Colcombet and Olivier Idir.
\newblock An algebraic characterisation of eve-positional languages.
\newblock ArXiv, link to appear, 2026.

\bibitem{ColcombetN06}
Thomas Colcombet and Damian Niwi\'{n}ski.
\newblock On the positional determinacy of edge-labeled games.
\newblock {\em Theor. Comput. Sci.}, 352(1-3):190--196, 2006.
\newblock URL: \url{https://doi.org/10.1016/j.tcs.2005.10.046}, \href {https://doi.org/10.1016/J.TCS.2005.10.046} {\path{doi:10.1016/J.TCS.2005.10.046}}.

\bibitem{Safra_flowers}
Thomas Colcombet and Konrad Zdanowski.
\newblock A tight lower bound for determinization of transition labeled b{\"u}chi automata.
\newblock In Susanne Albers, Alberto Marchetti-Spaccamela, Yossi Matias, Sotiris Nikoletseas, and Wolfgang Thomas, editors, {\em Automata, Languages and Programming}, pages 151--162, Berlin, Heidelberg, 2009. Springer Berlin Heidelberg.
\newblock URL: \url{https://doi.org/10.1007/978-3-642-02930-1_13}.

\bibitem{Emerson_Jutla_parity}
E.A. Emerson and C.S. Jutla.
\newblock Tree automata, mu-calculus and determinacy.
\newblock In {\em [1991] Proceedings 32nd Annual Symposium of Foundations of Computer Science}, pages 368--377, 1991.
\newblock \href {https://doi.org/10.1109/SFCS.1991.185392} {\path{doi:10.1109/SFCS.1991.185392}}.

\bibitem{LAR}
Yuri Gurevich and Leo Harrington.
\newblock Trees, automata, and games.
\newblock In {\em Proceedings of the Fourteenth Annual ACM Symposium on Theory of Computing}, STOC '82, page 60–65, New York, NY, USA, 1982. Association for Computing Machinery.
\newblock \href {https://doi.org/10.1145/800070.802177} {\path{doi:10.1145/800070.802177}}.

\bibitem{RabinPosi}
N.~Klarlund.
\newblock Progress measures, immediate determinacy, and a subset construction for tree automata.
\newblock In {\em [1992] Proceedings of the Seventh Annual IEEE Symposium on Logic in Computer Science}, pages 382--393, 1992.
\newblock URL: \url{https://doi.org/10.1016/0168-0072(94)90086-8}, \href {https://doi.org/10.1109/LICS.1992.185550} {\path{doi:10.1109/LICS.1992.185550}}.

\bibitem{Kopczynski06}
Eryk Kopczy\'{n}ski.
\newblock Half-positional determinacy of infinite games.
\newblock In Michele Bugliesi, Bart Preneel, Vladimiro Sassone, and Ingo Wegener, editors, {\em Automata, Languages and Programming, 33rd International Colloquium, {ICALP} 2006, Venice, Italy, July 10-14, 2006, Proceedings, Part {II}}, volume 4052 of {\em Lecture Notes in Computer Science}, pages 336--347. Springer, 2006.
\newblock \href {https://doi.org/10.1007/11787006\_29} {\path{doi:10.1007/11787006\_29}}.

\bibitem{Kopczynski07}
Eryk Kopczy\'{n}ski.
\newblock Omega-regular half-positional winning conditions.
\newblock In Jacques Duparc and Thomas~A. Henzinger, editors, {\em Computer Science Logic, 21st International Workshop, {CSL} 2007, 16th Annual Conference of the EACSL, Lausanne, Switzerland, September 11-15, 2007, Proceedings}, volume 4646 of {\em Lecture Notes in Computer Science}, pages 41--53. Springer, 2007.
\newblock \href {https://doi.org/10.1007/978-3-540-74915-8\_7} {\path{doi:10.1007/978-3-540-74915-8\_7}}.

\bibitem{Borel_determinacy}
Donald~A. Martin.
\newblock Borel determinacy.
\newblock {\em Annals of Mathematics}, 102(2):363--371, 1975.
\newblock URL: \url{http://www.jstor.org/stable/1971035}.

\bibitem{Ohlmann_universal_graphs}
Pierre Ohlmann.
\newblock Characterizing positionality in games of infinite duration over infinite graphs.
\newblock {\em TheoretiCS}, Volume 2, jan 2023.
\newblock URL: \url{http://dx.doi.org/10.46298/theoretics.23.3}, \href {https://doi.org/10.46298/theoretics.23.3} {\path{doi:10.46298/theoretics.23.3}}.

\bibitem{Rabin1968DecidabilityOS}
Michael~O. Rabin.
\newblock Decidability of second-order theories and automata on infinite trees.
\newblock {\em Bulletin of the American Mathematical Society}, 74:1025--1029, 1968.
\newblock URL: \url{https://api.semanticscholar.org/CorpusID:6015948}.

\bibitem{safra_determinization_88}
S.~Safra.
\newblock On the complexity of omega -automata.
\newblock In {\em Proceedings of the 29th Annual Symposium on Foundations of Computer Science}, SFCS '88, page 319–327, USA, 1988. IEEE Computer Society.
\newblock \href {https://doi.org/10.1109/SFCS.1988.21948} {\path{doi:10.1109/SFCS.1988.21948}}.

\bibitem{determinising_parity}
Sven Schewe and Thomas Varghese.
\newblock Determinising parity automata.
\newblock pages 486--498, 2014.

\bibitem{Seidl_Niwinski_1999}
Helmut Seidl and Damian Niwi\'{n}ski.
\newblock On distributive fixed-point expressions.
\newblock {\em RAIRO - Theoretical Informatics and Applications}, 33(4–5):427–446, 1999.
\newblock \href {https://doi.org/10.1051/ita:1999101} {\path{doi:10.1051/ita:1999101}}.

\end{thebibliography}
	
	\section{Tile monoids and Eve-positionality of ordered Büchi languages}\label{sec:eve-posi}
	

\knowledgenewrobustcmd{\om}{\cmdkl{\omega}}
\knowledgenewrobustcmd{\up}{\cmdkl{\uparrow}}
\knowledgenewrobustcmd{\ci}{\cmdkl{\circ}}

In this section, we give some results over "ordered Büchi automata", before establishing that the "ordered Büchi automata" recognise exactly the "Eve-positional" "$\omega$-regular languages".

\subsection{Properties of ordered Büchi tiles}

We observed that when we set a total order ${\stateleq}$ over $Q$, and consider the set $\Sii$ of upwards-closed tiles for the corresponding order ${\transleq}$, the set $\Sii$ still has a monoid structure, by \Cref{lem:Sii-monoid}.

We can then observe that for all "ordered Büchi automaton" $\A = (Q,\stateleq, I, \Gamma)$, a tile $t\in \Gamma$ is such that $t^\omega \in \LangS(\A)$ if and only if it admits a "horizontal" "Büchi transition" over a state in $I$.
This result will notably be useful when reasoning on "ultimately-periodic" words, typically for the "local preference properties".

\begin{restatable}{lemma}{lemomrejecting}\label{lem:tile-om-rejecting-bis}
	Let $\A = (Q, \stateleq, I, \Gamma)$ be an "ordered Büchi automaton".
	For all $t \in \Gamma$, $t^{\omega}\in \LangS(\A)$ if and only if there exists $q\in I$ such that $q \trans{t:0} q$.
\end{restatable}
\begin{proof} 
	If there is some $q\in I$ such that $q \trans{t:0} q$, immediately $t^\omega \in \LangS(\A)$, as shown by the run $((q,0,q))_{i\in \NN}$.

	Conversely, we assume that there is no such "horizontal" "Büchi transition". Notably, there is no transition $\delta$ in $t$ of the shape $p_I \trans{t} q$, with $p_I\in I$ and $q \notin I$. Indeed, if it were the case, by the downward closeness of $I$, $p_I \statel q$, and thus, denoting $\delta' := (p_I, 0, p_I)$, we have $\delta \translt \delta'$. This would imply, by upward closure, that $\delta'\in t$: contradiction. Similarly, we get that for all $(q_I,c,q')\in t$ with $q_I\in I$, $q'$ is such that $q' \stateleq q_I$.
	Let $( (q_i, r_i, q_{i+1}) )_{i \in \NN}$ be a "run@@tile" over $t^\omega$, with $q_0\in I$. By immediate recurrence using the previous result, for all $i\in \NN$, $q_i \in I$ and $q_{i+1} \stateleq q_i$. This means that the sequence $(q_i)_{i\in\NN}$ is ultimately constant, equal to some $q_\infty$. Since $q_\infty\trans{\ti}q_\infty$ is not a "Büchi transition", the run is "rejecting@@tile".
\end{proof}

We also easily observe that the different states in $Q$ directly define the residuals of the corresponding language, and that due to the structure imposed by ${\stateleq}$, these residuals are totally ordered.

\begin{lemma}\label{lem:ordered-residuals}
	Let $\A = (Q, \stateleq, I, \Gamma)$ be an "ordered Büchi automaton". Then the "residuals" of $\LangS(\A)$ are exactly the languages $\LangS(\A_U)$, where $\A_U = (Q, \stateleq, U, \Gamma)$ for $U$ a downwards-closed subset of $Q$. Furthermore, these residuals are totally ordered for the inclusion.
\end{lemma}
\oii{Lemme trop trivial selon le reviewer C. Et… en vrai il est un peu simple, mais pas forcément pour tout niveau de familiarité ?}
\begin{proof}
	We first prove that each residual corresponds to such a $\A_U$. Let $u^{-1}\LangS(\A)$ be a "residual" of $\LangS(\A)$ with $u$ a finite word. Let $p_u$ be the highest state (for ${\stateleq}$) such that $q_i \trans{u} p_u$ with some $q_i\in I$, and $U$ be the downwards-closed subset of $Q$ with maximal element $p_u$. Let us prove that $u^{-1}\LangS(\A) = \LangS(\A_U)$.
	
	For $w\in u^{-1}\LangS(\A)$, $uw\in \LangS(\A)$, witnessed by an "accepting run@@tile" $\rho_u \rho_w$ where $\rho_u$ is a path over $u$ and $\rho_w$ an accepting infinite path over $w$.  We then observe that for $(q_j, c_j, q_{j+1})$ the first transition of $\rho_w$, by upward closure of $w_0 \in \Sii$, $(p_u, c_j, q_{j+1}) \in w_0$. Therefore $w\in \Lang(\A_{U})$.
	
	Conversely, for $w\in \Lang(\A_{U})$, by definition of $p_u$, there exists a path $\rho_u$ of the shape $q_i \trans{u} p_u$. Therefore, for $\rho_w$ "accepting run@@tile" over $w$ in $\A_{U}$, we similarly observe that $\rho_w$ starts in a vertex $q_w \in U$, therefore $q_w \leq p_u$. We thus obtain, by upward closure of $w_0$, that there exists an "accepting run@@tile" $\rho'_w$ over $w$ starting in $p_u$. Therefore, $\rho_u \rho'_w$ is an "accepting run@@tile" in $\A$. Therefore  $u w \in \LangS(\A)$, and thus $w \in u^{-1}\LangS(\A)$.
	
	We finally observe that these residuals are totally ordered by inclusion as, by the same upward-closure argument, it is always possible for a run starting in $q$ to mimic a run starting in some $q' \stateleq q$.
\end{proof}

	\subsection{Ordered Büchi languages are Eve-positional}
	
	The goal of this section is to prove the following theorem:
	
	\oBiffEveposi
	
	We do so by double inclusion, making use of the previous structural results on "ordered Büchi automata". We first exhibit that the "ordered Büchi" languages are "Eve-positional", as they satisfy the "local preference properties".
	
	\oBEveposi
	\begin{proof}
	Let $\A = (Q, \stateleq, I, \Gamma)$ be an "ordered Büchi automaton", let us prove that $\LangS(\A)$ satisfies the "local preference properties".
	In this proof, for $u\in \Gamma^*$ a finite word, as in the proof of \Cref{lem:ordered-residuals}, we define $p_u$ as the highest state in $Q$ (with respect to ${\stateleq}$) reachable from $I$ after reading $u$. We also denote $U$ as the downwards-closed subset of $Q$ of maximal element $p_u$, and $\A_U := (Q, \stateleq, U, \Gamma)$.
	\begin{itemize}
		\item This is a direct consequence of \Cref{lem:ordered-residuals}, as for two finite words $u,u'$, it suffices to look at which is the greatest residual $u^{-1}\LangS(\A)$ and $u'^{-1}\LangS(\A)$. If it is $u^{-1}\LangS(\A)$, we then observe that for $w'$ such that $u' w'\in \LangS(\A)$, immediately $w' \in u^{-1}\LangS(\A)$ and thus $u w' \in \LangS(\A)$. The other case is treated symmetrically.
		
		\item Let $u\in \Gamma^*$, $v \in \Gamma^+$ and $w\in \Gamma^\omega$ be words such that $uvw\in \LangS(\A)$. If $v^\omega\in \LangS(\A_U)$, we immediately observe that $uv^\omega \in \LangS(\A)$. Else, by \Cref{lem:tile-om-rejecting-bis}, there is no $p_v \in U$ such that $p_v \trans{v:0} p_v$, and notably we do not have $p_u \trans{v:0} p_u$. By \Cref{lem:Sii-monoid}, up to considering $t_v$ the "tile" obtained by "product" of the successive "tiles" of $v$, $t_v \in \Sii$. Therefore, the highest reachable state from $U$ after reading $v$ is at most $p_u$: else, a transition $p_u \trans{t_v:0} p_u$ would exist by upward closure of $t_v$, contradiction. By transitivity, the highest reachable state from $I$ after reading $uv$ is therefore at most $p_u$. Thus, as $uvw \in \LangS(\A)$, $w\in \LangS(\A_U)$, and therefore $uw\in \LangS(\A)$.
		
		\item Let $u\in \Gamma^*$ and $v,v' \in \Gamma^+$ be such that $u(vv')^\omega \in \LangS(\A)$. Then $(vv')^\omega\in \LangS(\A_U)$, and therefore, by \Cref{lem:tile-om-rejecting-bis}, there exists $q\in U$ such that $q \trans{vv' : 0} q$. Let us show that either $q \trans{v:0}q$ or $q \trans{v':0}q$, which respectively imply that $uv^\omega\in \LangS(\A)$ or $uv'^\omega\in\LangS(\A)$. Let $q'$ be the highest reachable state in $Q$ from $q$ after reading $v$ (with some priority $c$). If $q \statel q'$ (third case in \Cref{fig:Eve-posi-PI}), then by upward closure, $(q,c,q') \translt (q,0,q)$ and thus $q \trans{v:0} q$. Else, if $q' \statel q$ (fourth case in \Cref{fig:Eve-posi-PI}), necessarily $q' \trans{v'} q$, and thus by upwards-closure, $q \trans{v':0} q$. In the remaining case $q = q'$, and thus at least one of $q \trans{v:0}q$ or $q \trans{v':0}q$ holds (two first cases in \Cref{fig:Eve-posi-PI}), which concludes the proof.

		\end{itemize}
	\end{proof}

	\begin{corollary}\label{cor:oB-Eveposi}
		All "ordered Büchi languages" are "Eve-positional".
	\end{corollary}
	\begin{proof}
		Let $L$ be such a language over an alphabet $\Pi$, associated with an "ordered Büchi automaton" $\A = (Q, \stateleq, I, \Gamma)$ with a morphism $f: \Pi \to \Gamma$. Then, as $\LangS(\A)$ is "Eve-positional", it notably satisfies the "local preference properties". Therefore, looking at $L$ directly, we observe that these properties also hold:
		\begin{asparaitem}
			\item For all $u,u'\in\Pi^*$, and $w,w'\in\Pi^\omega$,
			if $uw\in L$ and $u'w'\in L$, then similarly $f(uw)\in \LangS(\A)$ and $f(u'w')\in \LangS(\A)$. Therefore, as $\LangS(\A)$ is "Eve-positional", either $f(uw')\in \LangS(\A)$ or $f(u'w)\in \LangS(\A)$. Hence, by definition of $f$, either $uw'\in L$ 	or $u'w\in L$.
			\item For all $u\in\Pi^*$, $v\in \Pi^+$, and $w\in\Pi^\omega$, if $uvw\in L$, similarly $f(uvw) = f(u)f(v)f(w)\in \LangS(\A)$, and as $\LangS(\A)$ is "Eve-positional", by the second "local preference property", either $f(u)f(v)^\omega \in \LangS(\A)$, or $f(u)f(w)\in \LangS(\A)$. Therefore, by definition of $f$, either $uv^\omega \in L$ or $uw\in L$.
			\item For all $u\in\Pi^*$, $v,v'\in \Pi^+$, if $u(vv')^\omega \in L$, similarly $f(u(vv')^\omega) = f(u)(f(v)f(v'))^\omega\in \LangS(\A)$, 
		and as $\LangS(\A)$ is "Eve-positional", by the third "local preference property", either $f(u)f(v)^\omega \in \LangS(\A)$, or $f(u)f(v')^\omega\in \LangS(\A)$. 
		Therefore, by definition of $f$, either $uv^\omega \in L$ or $uv'^\omega\in L$.
	\end{asparaitem}
	Therefore $L$ satisfies the "local preference properties", and is thus "Eve-positional".
	\end{proof}
	
	We now show the converse implication, that is, all "Eve-positional" "$\omega$-regular" languages are "ordered Büchi languages". To do so, we first show that the words admitting an accepting run in an "$\eps$-complete" automaton can be chosen of a specific shape, before building an "ordered Büchi automaton" mirroring the behaviour of such an automaton.
	
	\begin{lemma}\label{lem:eps-complete-intertwined}
		Let $A$ be an "$\eps$-complete automaton". For every word $w = (w_i)_{i\in \NN}\in \Lang(A)$, there exists an accepting run in $A$ over the word $(\eps w_i \eps)_{i\in \NN}$.
	\end{lemma}
\begin{proof}
	We denote by $J = \inter{2d+2}$ the "index" of $A$.
	As $A$ is an "automaton with $\eps$-transitions", there exists some word $w'\in \Pieps^\omega$, with an accepting run in $A$, such that $w'$ with its $\eps$'s removed is $w$.
	We first observe that the $\eps$-transitions are closed by letter concatenation: for all $p,q,r \in Q_A$ and $c,c' \in J$, if $p \trans{\eps:c} q \trans{\eps:c'}r$, then $p\trans{\eps : \min(c,c')} r$. If $c= c'$, the result is immediate by transitivity of $\trans{\eps:c}$. Else, if $c < c'$, we reason by case disjunction on the parity of $c$.
	If $c$ is even, we cannot have that $r \trans{\eps:c+1} p$: else, by transitivity of $\trans{\eps:c}$ and due to the fact that $\trans{\eps:c+1}$ is refined either by $\trans{\eps:c'}$ or $\trans{\eps:c'+1}$ (depending on which one is odd), we would get that $q \trans{\eps:c+1} r \trans{\eps:c} q$, and thus that $q \trans{\eps:c} q$ (as $\trans{\eps:c}$ is the strict variant of $\trans{\eps:c+1}$): contradiction with the antisymmetry of $\trans{\eps:c}$. Therefore, by definition of $\trans{\eps:c}$, $p \trans{\eps:c} r$.
	If $c$ is odd, then as $\trans{\eps:c}$ is refined either by $\trans{\eps:c'}$ or $\trans{\eps:c'+1}$ (depending on which one is odd, denoted $c''$), as in either case $q \trans{\eps:c''} r$, then $q \trans{\eps:c} r$, which concludes by transitivity of $\trans{\eps:c}$.
	The case $c' < c$ is treated symmetrically.
	
	Therefore, if in $w'$ there exist multiple consecutive $\eps$, we can replace them by a single occurrence of $\eps$ without changing the existence of an "accepting run@@aut".

	Let us now show that we can insert an $\eps$ in $w'$ before all letters whilst maintaining this property. It suffices to observe that any such run is accepting with some even priority $2k \leq 2d$, yet it is always possible to insert an $\eps$ self-loop of priority $2d+1$ between each step of the run. Doing so does not affect the minimal infinitely recurring priority. Therefore, the word obtained by inserting in $w'$ an $\eps$ before all the letters still admits an accepting run.
	
	Combining these results, we get that $A$ admits an accepting run over $(\eps w_i \eps)_{i\in \NN}$.
\end{proof}

Then, for $L$ an "Eve-positional" "$\omega$-regular" language, according to \Cref{theorem:co}, any "Büchi automaton" recognising $L$ is "$\eps$-completable". It therefore suffices to design an "ordered Büchi automaton" that mimics the behaviour of some such "automaton" (which exists, as $L$ is "$\omega$-regular").

\BuchitooB
\begin{proof}
	In order to prove that $L$ is "ordered Büchi", we need to build an "ordered Büchi automaton" $\A = (Q,\stateleq, I, \Gamma)$, along with a morphism $f:\Pi \to \Gamma$, such that a word $w\in \Pi^\omega$ is in $L$ if and only if $f(w) \in \LangS(\A)$.
	
	Let $A$ be an "$\eps$-completable" "Büchi automaton" recognising $L$. We consider its "$\eps$-completion" $A'$, which also recognises $L$, and the total preorder $\trans{\eps:1}$ that this implies over $Q_A$. We build $Q$ as the quotient of $Q_A$ over the equivalence classes defined by $\trans{\eps:1}$ -- which thus defines a total order over $Q$. By "$\eps$-completeness", $I_A$ can be assumed downwards-closed for the pre-order $\trans{\eps:1}$, and we thus define $I$ as the downward-closed subset of $Q$ such that all its classes have one element in $I_A$ (and thus have all their elements in $I_A$). Let us now build the desired morphism $f : \Pi \to \Sii$.
	To do so, we first build $f_0:\Pieps\to \Si$, before defining $f$ that maps any letter $a$ to the upward-closure of $f_0(a)$ for ${\transleq}$. That is, $f_0(a)$ is a generator of $f(a)$.
	
	For $a\in \Pieps$, the "tile" $f_0(a)$ is defined as follows: for all $p,s \in Q$ equivalence classes, $p \trans{f_0(a):1} s$ if there exists $p_A\in p, s_A \in s$ such that $p_A \trans{a:1} s_A$, and similarly for $p \trans{f_0(a):0} s$.
	We first remark that this operation on the letter $\eps$ produces exactly the tile $\ti$, as $\eps$-transitions of priority 1 describe the total order ${\stateleq}$, and $\eps$-transitions of priority 0 describe its strict variant. 
	We then observe that $f$ mapping each $a$ to the upward-closure of $f_0(a)$ is indeed a function from $\Pi$ to $\Sii$, which we extend into a morphism over the free word monoid in the standard fashion.
	We can now prove that for all $w\in \Pi^\omega$, $w\in L$ if and only if $f(w)\in \LangS(\A)$.
	
	Let $w = (w_i)_{i\in \NN}\in \Pi^\omega$. If $w\in L$, by \Cref{lem:eps-complete-intertwined}, the word $w' := (\eps w_i \eps)_{i\in \NN}$ admits an "accepting run@@aut" consisting, for each $i$, of a left transition $\delta^{\eps}_{l,i}$, a letter-transition $\delta_i$, and a right transition $\delta^{\eps}_{r,i}$, all in $\Delta_{A'}$. We denote the corresponding vertex sequence as $(q_{i,0},q_{i,1},q_{i,2})_{i\in \NN} \in (Q^3)^\omega$ where $q_{0,0}$ belongs to $I_A$. From there, by definition of $f_0$, we deduce the existence of a "Büchi path@accepting path" along the word $(f_0(\eps)f_0(w_i)f_0(\eps))_{i\in \NN} = (\ti f_0(w_i) \ti)_{i\in \NN}$.
	We then observe that for all $i\in \NN$ and corresponding $c'_i \in \inter{2}$ such that $q_{i,0} \trans{\ti} q_{i,1} \trans{f_0(w_i):c'_i} q_{i,2} \trans{\ti} q_{i+1,0}$, dominated by some priority $c_i$, we have that $(q_{i,1}, c'_i, q_{i,2}) \transleq (q_{i,0}, c_i, q_{i+1,0})$. Therefore, by definition of $f$, $\delta^f_i := (q_{i,0}, c_i, q_{i+1,0})\in f(w_i)$, and we thus easily build the corresponding accepting run $(\delta^f_i)_{i\in \NN}$ over $f(w)$, proving that $f(w)\in \LangS(\A)$.
	
	Conversely, if $f(w)\in \LangS(\A)$, there exists an "accepting run@@tile" $((q_i, c_i, q_{i+1}))_{i\in \NN}$ over $f(w)$. Let us show that there exists an "accepting run@@aut" in $A'$ over $w' := (\eps w_i)_{i\in \NN}$.
	By definition of $f$, for all $i$, there exists $\delta_i = (q_{i,0}, c_{i,0}, q_{i,1}) \in f_0(w_i)$ such that $\delta_i \transleq (q_i, c_i, q_{i+1})$ -- that is, $q_{i,0}\stateleq q_i$, and $q_{i+1} \stateleq q_{i,1}$. We thus observe that $q_{0,0} \stateleq q_0$, and for all $i$, $q_{i+1, 0} \stateleq q_{i,1}$.
	Therefore, by considering for each $i$ the states $p_i \in q_{i,0}$, $p'_i \in q_{i,1}$ as the witnesses of the existence of the transition $\delta_i$, we observe that for all $i$, $p_{i,0} \trans{\eps} p'_{i+1,0}$. We thus observe that $((p'_{i-1}, c_{i, \eps}, p_i) (p_i,c_{i,0}, p'_i))_{i\in \NN}$, with $p'_{-1} \in I_A$ and all the $c_{i,\eps}$ chosen minimal for the corresponding vertex pair, is a "run@@aut" in $A'$ over $(\eps w_i)_{i\in \NN}$. We finally observe that it is accepting by case disjunction: either infinitely often $c_{i,0} = 0$, in which case the run is immediately accepting, or, as $((q_i, c_i, q_{i+1}))_{i\in \NN}$ is accepting, infinitely often $c_i = 0$, and thus infinitely often $c_i < c_{i,0}$. At these indices, necessarily $q_{i,0} \statel q_i$ or $q_{i+1} \statel q_{i,1}$. Therefore, either the $\eps$-transition from $p'_{i-1}$ to $p_i$ is over two different classes, or it is the case for the one from $p'_i$ to $p_{i+1}$. That is, respectively $c_{i, \eps}$ or $c_{i+1, \eps}$ is equal to $0$, by "$\eps$-completeness" (and as $c_{i, \eps}$ is always chosen minimal). Therefore $w\in \Lang(A') = L$.
\end{proof}

It is even possible to obtain a stronger version of this result, with a transformation from an "$\eps$-complete" "non-deterministic parity automaton" to an "ordered Büchi automaton" recognising the same language (once again, up to a renaming of alphabet). 
This transformation incurs no additional state blow-up beyond what is already needed when transforming "non-deterministic parity automata" into "non-deterministic Büchi automata". This result, however, comes with a more elaborate proof, detailed in the appendix:

\NPAtooB

We thus obtain the missing implication in \Cref{theorem:co}: for $L$ an "$\omega$-regular" language recognised by an "$\eps$-completable" "non-deterministic automaton", by \Cref{thm:NPA-to-oB}, $L$ is "ordered Büchi". Therefore, by \Cref{cor:oB-Eveposi}, $L$ is indeed "Eve-positional".

\epscompl

	\section{Determinization of Eve-positional languages}\label{sec:det}
	

In this section, we present a transformation from an "ordered Büchi automaton" to a "deterministic parity automaton" recognising the same language. This transformation follows ideas from the last appearance record technique \cite{LAR} and has at most a factorial blow-up. We observe that it produces an automaton of minimal size, provided that the tile alphabet contains some expressive enough subset.

\subsection{Determinization from ordered Büchi to parity}

\AP For a tile $t\in \Sii$ and a subset $X \subseteq Q$, we denote $t(X) := \{q \mid \exists p \in X, p \trans{t} q\}$. That is, $t(X)$ represents all the states reachable from $X$ after reading $t$. Note that by upward closure of $t$ for ${\transleq}$, $t(X)$ is a downwards closed set.
We additionally define, for a tile $t\in \Sii$, $\intro*\ska : Q \to Q \uplus \bot$ as the function such that $\ska(q) = \max_{\stateleq}\{q' \mid q \trans{t} q'\}$, with $\ska(q) = \bot$ if no such $q'$ exists (by convention $\bot$ is inferior to all the elements in $Q$). We immediately observe that $\ska$ is monotone with respect to ${\stateleq}$, once again by upward closure of $t$.

\subsubsection{Definition of the deterministic automaton}

\AP
Let $\A = (Q,\stateleq, I, \Gamma)$ be an "ordered Büchi automaton". The first goal of this section is to describe a "deterministic parity automaton"
\phantomintro{\Ai}
\begin{align*}
	\reintro*\Ai = ( S,\Gamma,s_{I},\De )
\end{align*}
that "recognises@@aut" the language $\LangS(\A)$ (see \Cref{lem:Ai-complete} and \Cref{lem:Ai-correct}). Its behaviour is quite complex, but intuitively it aims to keep track, for each reachable state $q$ in $Q$, of the best feasible run ending in $q$. We refer to the \Cref{ex:det} for a better understanding.

The set of ""records"" $S$ is the set of tuples $(s_0,\dots,s_k)$ such that 
\begin{asparaenum}
	\item $\{s_0,\dots,s_k\}$ is a downward closed subset of~$Q$ for ${\stateleq}$, 
	\item the $s_i$'s are pairwise distinct, and
	\item $s_0$ is maximal for ${\stateleq}$ among the $s_i$'s.
\end{asparaenum}
Equivalently, a record is an injective map $s:\inter{k} \to Q$ (for some $k \leq n$) with a downward closed image.

Intuitively, the current "record" remembers what the different accessible vertices in $Q$ are and the relative age of the corresponding subpaths, with older paths occupying smaller indices. These paths will then be extended along $\ska$, for $t$ the tile read. The exact rules for the fusion of paths and creation of new paths are detailed below.

	
The "initial state" $s_I$ is the "record" whose image contains exactly $I$, such that $s(0) = \max_{\stateleq}(I)$ and the following positions contain the states of~$I$ in decreasing order over ${\stateleq}$.

\AP The transition function $\intro*\De : S \times \Sii \to [-1,2n-1] \times S$ is defined for all $s = (s(0),\dots,s(k)) \in S$ and $t \in \Sii$ to be
\[
\De(s,t) = (c,s')
\]
in which $c$ and $s'$ are defined as follows. For all $i\in\dom(s)$, we define
\begin{align*}
	\intro*\best(i)&:=\ska(s(i)),&
	\intro*\leader(i)&:=\min_{\leqslant}\{j\in\dom(s)\mid \best(i)=\best(j)\}\\
	P&:=\img(\leader)\subseteq \inter{k},&
	R&:=t(\img(s)) \setminus \best(P).
\end{align*}
For $i$ an index in $\dom(s)$, $\leader(i)$ thus corresponds to the smallest index of $\dom(s)$ that has same image as $i$ by $\ska(s) = \best$. As the $s(i)$'s correspond to the end vertices of different concurrent paths, $\leader(i)$ thus corresponds to the oldest path that has for image $\best(i)$. We denote $P$ the set of leaders, as they will correspond to ``preserved'' paths. The indices in $R$, conversely, correspond to ``reset'' paths in $s'$.
\AP We say that an index $i\in \dom(s)$ is ""forgotten"" if $i\neq \leader(i)$ (that is, $i \in \inter{k} \setminus P$).

Let~$b$ be the unique increasing map from $\inter{\,|P|\,}$ to $P$, and let $f$ be the unique decreasing map from $\inter{\,|R|\,}$ to $R$ \footnote{$f$ could be any bijection from $[|R|]$ to $R$, but choosing consistently an order ensures a smaller resulting automaton}. We define the "record" $s'$ to have domain $\inter{\,|P|+|R|\,} = \inter{\,|t(\img(s))|\,}$, image $ t(\img(s))$ and to be such that
\begin{align*}
	\text{for all~}i\in \inter{\,|P|+|R|\,}\text{,}\quad
	s'(i):=\begin{cases}
		\best(b(i))&\text{for}~0\leqslant i < |P|\\
		f(i-|P|)&\text{otherwise.}
	\end{cases}
\end{align*}
Note that if $P \neq \emptyset$, $s'(0) = \best(b(0)) = \best(0) = \ska(s(0)) = \max_{\stateleq} t(\img(s))$, and this state is well-defined (if $P = \emptyset$, $s'$ is the empty "record").
Intuitively, all the paths from indices in $P$ are preserved, in their initial order. The indices after $|P|$ then correspond to brand new paths (in $R$), assigned in decreasing order.

\AP The "priority" of this transition is defined as $c:=\min(2\green,2\red-1)$ where\phantomintro{\green}\phantomintro{\red}
\begin{align*}
	\reintro*\green&:=\min_\leqslant\{i\mid s(i)\trans{t:0}s'(i)\}&
	\text{and}\reintro*\quad\red&:=\min_\leqslant(\dom(s)\setminus P)\ ,
\end{align*}
with, in each case, the convention that the minimum over the empty set is $|t(\img(s))|$.
That is, intuitively, we want to accept with priority $2i$ if eventually the path at index $i$ is never reset, and witnesses infinitely many Büchi transitions. Conversely, if we are rejecting with priority $2i+1$, then all the paths at indices $<i$ keep going infinitely but eventually no longer see any Büchi transition, and all the other paths are eventually "forgotten".

\begin{example}\label{ex:det}
	Let $Q=\{q_0, q_1, q_2, q_3, q_4, q_5\}$, with ${\stateleq}$ corresponding to the standard order over the states' index. We observe here a transition in $\Ai$ over some "tile" $t\in \Sii$ whose skeleton is described in \Cref{fig:transition-leader}.
	More precisely, we are looking at a transition from a given "record", here $(q_5, q_3, q_4, q_0, q_2, q_1)$, over the tile $t$. The first step is then to establish, for each index, its leader. Here, the set $P$ is thus $\{0,1,3\}$, as shown in \Cref{fig:transition-leader}.
	\begin{figure}[!htb]
		\begin{minipage}[b]{0.3\textwidth}
			\includegraphics[height=5cm]{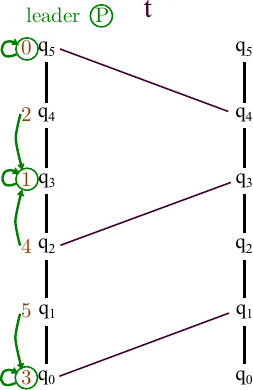}
			\caption{We compute the function $\leader$ from the given "record" $s$ in brown. This defines the set $P$.}
			\label{fig:transition-leader}
		\end{minipage}
		\hfill
		\begin{minipage}[b]{0.3\textwidth}
			\includegraphics[height=5cm]{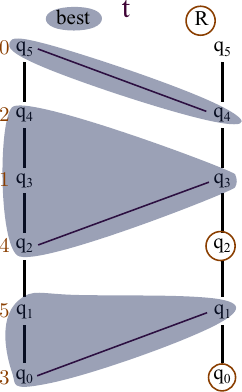}
			\caption{We regroup the indices that have the same image by $\best$, and thus by $\ska$. This defines the set $R$.}
			\label{fig:transition-best}
		\end{minipage}
		\hfill
		\begin{minipage}[b]{0.3\textwidth}
			\includegraphics[height=5cm]{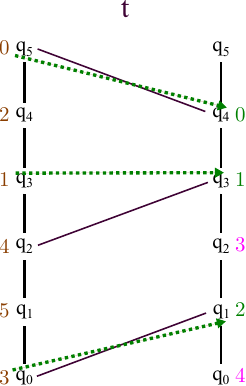}
			\caption{(c) We now compute $s'$, with the help of the functions $b$ and $f$. \,}
			\label{fig:transition-b}
		\end{minipage}
	\end{figure}
	We then consider the image $\best(P) = \{q_1, q_3, q_4\}$, which allows us to define the set $R = \{q_0, q_2\}$ as its complement in $t(\img(s))$, as shown in \Cref{fig:transition-best}.
	We finally look at the unique increasing function $b$ from $\inter{\,|P|\,} = \inter{3}$ to $\{0,1,3\}$. We can thus describe $s'(\inter{2})$ as the different $\ska(b(i))$, as shown in \Cref{fig:transition-b} with the green arrows. The remaining positions in the "record" are assigned the remaining states in the downwards closed subset, in decreasing order, as visible in magenta in \Cref{fig:transition-b}.
\end{example}

Whenever two paths would fuse (due to having the same image by $\ska$), we keep the oldest one -- the one of smallest index in $s$. Therefore, all the indices $i\notin P$ are "forgotten", as they describe paths which are not as relevant as the ones they fuse into, and the states in $R$ are now considered as the starting points of brand new paths.\\
The priority choice mirrors this intuition: it outputs an odd priority
corresponding to the oldest path forgotten, or an even priority
corresponding to the oldest path witnessing a Büchi transition. Thus, an
even priority $2i$ ultimately dominating implies that the index $i$
is never "forgotten" past some point (nor all the smaller indices), and the infinite path it describes witnesses an infinity of Büchi transitions.

Note that this automaton can have at most $1 + \sum_{i=0}^{n-1} i!$ states : $1$ for the empty record, then, for each non-empty record $s$, $s(0)$ is fixed and there are thus $(|\img(s)|-1)!$ possible such records. In many cases, however, they are not all accessible. As we consider our automata as trimmed to their accessible states by default, we then consider $\Ai$ to be smaller. This idea will be detailed in the \Cref{sec:optim}
	
\subsubsection{Correction of the recognised language}
	
	Let us now prove that $\Ai$ recognises $\LangS(\A)$. By determinism, for $w\in \Sii^*$ and $s\in S$, we denote as $\De(s,w)$ the "record" $s'$ reached after reading $w$ from $s$.
	
	\begin{lemma}\label{lemma:reachable-states}
		Let $w \in \Sii^*$ be a finite word and $s := \De(s_I,w)$ be the state reached in $\Ai$ after reading $w$. Then $\img(s)$ is the set of "states" of $Q$ reachable by "partial runs over" $w$ starting in $I$.
	\end{lemma}
	\begin{proof}
		The proof proceeds by induction on the length of $w$.
		We have that $\img(\De(s_I,\io))=\img(s_I) = I$ and the result holds.
		Assume now that the result holds for  $w\in \Sii^*$, and consider the word $w t$ for some $t\in\Sii$.
		As made explicit in the definition of $\De$, for all $s\in S$ and all~$t\in\Sii$, $\img(\De(s,t))=t(\img(s))$. 
		Let~$X$ be the set of states reachable from an initial state by "partial runs" over~$w$.
		We have by induction hypothesis that $X=\img(\De(s_I,w))$.
		Hence the set of states reachable after reading
		$w t$ is $t(X)=t(\img(\De(s_I,w)))=\img(\De(\De(s_I,w),t))=\img(\De(s_I,wt))$.
	\end{proof}
	
	We can now prove the double inclusion stating that $\Ai$ recognises exactly the accepting tile sequences.
	
	\begin{lemma}\label{lem:Ai-correct}
		For $\A = (Q,\stateleq, I, \Gamma)$ an "ordered Büchi automaton", and $\Ai$ obtained by the previous determinization procedure applied on $\A$,
		$L(\Ai)\subseteq \LangS(\A)$.
	\end{lemma}
	\begin{proof} 
		Let $((s_j,w_j,c_j,s_{j+1}))_{j\in\NN}$ be an "accepting run@@aut" of $\Ai$ over a word~$w$, we intend to show that $w\in \LangS(\A)$.
		Let $\best[w_j], b_j$, $P_j$, $\green_j$ and $\red_j$ be as defined in the construction of $\Ai$ for the transition $(s_j,w_j,c_j,s_{j+1})$.
		Let $2g$ be the minimal even priority that occurs infinitely in ~$(c_j)_{j\in \NN}$, and finally let~$j_0\in\NN$ be such that $c_j\geqslant 2g$ for all $j\geqslant j_0$.
		
		We claim first that for all $j\geqslant j_0$ and for all $i \leq g$, $i\in P_j$. 
		Indeed, since~$c_j\geqslant 2g$, we have $2\red_j-1\geqslant 2g$, and hence $\red_j> g\geqslant i$.
		Thus $i\not\in(\dom(s_j)\setminus P_j)$, that is~$i\in P_j$.
		
		We claim now that for all  $j\geqslant j_0$ and for all $i \leq g$, $s_j(i)\trans{w_j}s_{j+1}(i)$. 
		As $b_j$ is increasing, and by recurrence on $i \leq g$, we get that for all such $i$, $b_j(i) = i$.
		Therefore, $s_{j+1}(i)= \best[w_j](b_j(i)) = \best[w_j](i) = \ska[w_j](s_j(i))$, and thus $s_{j+1}(i) \in w_j(\{s_j(i)\})$.
		Hence, for all $i \leq g, s_j(i)\trans{w_j}s_{j+1}(i)$.
		
		We claim also that for~$j$ such that $c_j=2g$, $s_j(g)\trans{w_j:0}s_{j+1}(g)$. This is clear from the definition of $\green_j$.
		
		Hence, we can construct a "run@@tile" over $w$ as follows.
		Over the prefix of length~$j_0$, by \Cref{lemma:reachable-states}, we can use a "partial run" starting in an initial state and reaching~$s_{j_0}(g)$. Then our run continues using, for all~$j\geqslant j_0$, an existing transition $s_j(g)\trans{w_j}s_{j+1}(g)$ (that exists by the second claim) and infinitely often $s_j(g)\trans{w_j:0}s_{j+1}(g)$ (according to the third claim).
		This is a Büchi accepting run over~$w$.
	\end{proof}
	
	\begin{lemma} \label{lem:Ai-complete}
		For $\A = (Q,\stateleq, I, \Gamma)$ an "ordered Büchi automaton", and $\Ai$ obtained by the previous determinization procedure applied on $\A$,
		$ \LangS(\A)\subseteq \Lang(\Ai)$.
	\end{lemma}
	\begin{proof}
		Let $w=(w_j)_{j\in \NN}\in \LangS(\A)$ be a word with "accepting run@@tile" $(q_j, c_j, q_{j+1})_{j\in \NN}$. We denote as $\rho = (s_{j})_{j\in \NN}$ the unique run of $\Ai$ over $w$.
		
		\AP We denote in this proof a ""track"" as a possibly finite "run@@tile" $(\delta_j = (p_j, l_j, p_{j+1}))_{j_0 \leq j < \alpha}$, starting in a state $p_{j_0} \in s_{j_0}$, with $\alpha \in \NN \cup \{\omega\}$ maximal such that $\forall j \in [j_0, \alpha), p_{j+1} = \ska[w_j](p_{j})$. A "track" can thus be characterised by its starting index $j_0$ and its starting state in $s_{j_0}$. We observe by recurrence that for all $j \in J := [j_0, \alpha), p_j \in s_j$.
		We note that by maximality, $\alpha$ is finite if and only if there exists $j^*$ such that $\ska[w_{j^*}](p_{j^*}) = \bot$. We denote $(i_{p,j})_{j \in J}$ the sequence of indices such that $s_j(i_{p,j}) = p_j$. We observe that $(i_{p,j})_{j\in J}$ is non-increasing, and decreases only when a smaller index is "forgotten". 
		Therefore, in the case of an infinite "track" $(\delta_j)_{j \in J}$, the indices $(i_{p,j})_{j\in J}$ eventually stabilize in some index $i_p$. Past this point, we observe that no index $\leq i_p$ is "forgotten", as it would cause $i_{p,j}$ to decrease further, contradiction. Therefore, no odd priority smaller than $2i_p$ is produced. Hence, it suffices to exhibit an infinite "track" starting in $(j_0 = 0,p_0 \in I)$ with an infinity of Büchi transitions to prove that $\rho$ is accepting.
		
		Let $(\delta_j = (p_j, l_j, p_{j+1}))_{j\in J}$ a "track" starting in $(j_0 = 0, p_{0} \stategeq q_{0})$. By induction on $j$ and monotony of $\ska[w_j]$, we easily show that $\forall j, p_j \stategeq q_j$. Therefore, this "track" is infinite, and $(i_{p,j})_{j \in J}$ eventually stabilizes in some index $i_p$.
		We define two "tracks" to be equivalent if they are eventually equal. We observe that due to the determinism of $\ska[w_j]$, two "tracks" $(\delta_j = (p_j, l_j, p_{j+1}))_{j\in J}, (\delta'_j = (p'_j, l'_j, p'_{j+1}))_{j'_0\leq j < \alpha'}$ are equivalent if and only if there exists $j \in J \cap [j'_0, \alpha')$ such that $p_j = p'_j$. A consequence of this fact is that there are only finitely many equivalence classes of infinite tracks. They describe distinct infinite "tracks", that never intertwine as the different $\ska[w_j]$ are order-preserving. We can thus order them according to $\statel$.
		
		Let us consider a representative $(\delta_j = (p_j, l_j, p_{j+1}))_{j\geq j_0}$ of the minimal equivalence class of "tracks" above $(q_j){j \in \NN}$ (they are thus all infinite). We denote $k$ the time at which $(i_{p,j})_{j \in J}$ stabilizes in its final index $i_p$.

		\begin{figure}[!htb]
			\centering
			\includegraphics[height=5.8cm]{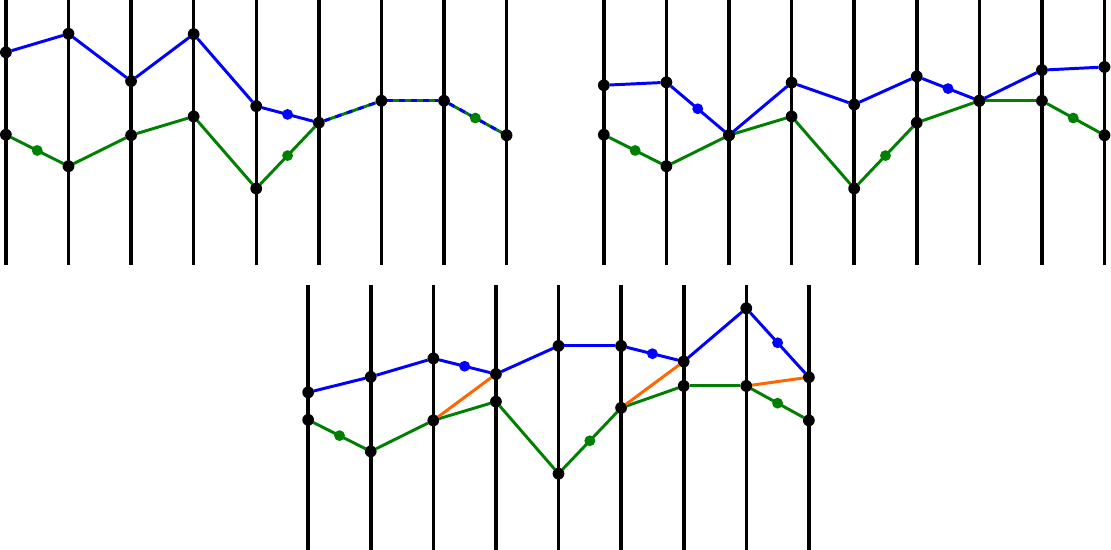}
			\caption{\label{fig:tracks-det} In these figures, the green run corresponds to the run $(q_j, c_j, q_{j+1})_{j\in \NN}$, and the blue run to the representative track $(\delta_j)_{j\in \NN}$. The following three situations can arise: either they eventually fuse, they meet infinitely often, or past some point they no longer meet. The existence of an infinity of orange transitions is a consequence of the minimality of $(\delta_j)_{j\in \NN}$. All the blue Büchi transitions represented are obtained by upward closure of the corresponding tiles.}
		\end{figure}
		
		If there exists some $k' \geq k$ such that for all $j \geq k', p_j = q_j$ (first case in \Cref{fig:tracks-det}), then as $(q_j, c_j, q_{j+1})_{j\in\NN}$ witnesses an infinity of Büchi transitions, so does $(\delta_{j})_{j\geq j_0}$.\\
		Else, if there exist infinitely many $(j_l)_{l \in \NN}$ such that $\forall l \in \NN, p_{j_{2l}} = q_{j_{2l}}$ and $p_{j_{2l+1}} \stateg q_{j_{2l+1}}$ (second case in \Cref{fig:tracks-det}): then for all $l \in \NN$, considering the maximal $j\geq j_{2l+1}$ such that $p_{j} \stateg q_{j}$, we have that $q_j \trans{w_j} q_{j+1}$, $p_j \trans{w_j} q_{j+1}$, and $p_j \stateg q_j$. Therefore, by upward closure of $w_j$ by ${\transleq}$, $p_j \trans{w_j:0} q_{j+1} = p_{j+1}$. Thus $(\delta_j)_{j \geq j_0}$ witnesses an infinity of Büchi transitions.\\
		Else, $\exists k' \geq k \in \NN, \forall j \geq k', p_j \stateg q_j$ (third case in \Cref{fig:tracks-det}). Let $j_1 \geq k'$ be an ulterior point in time, we consider the "track" $(\delta'_j = (p'_j, l'_j, p'_{j+1}))_{j \geq j_1}$ starting in $(j_1, q_{j_1})$. Then $(\delta'_j)_{j \geq j_1}$ can not go below $(q_j)_{j \in \NN}$ and is infinite. Therefore, by minimality of $(\delta_j)_{j \geq j_0}$, we get that $(\delta_j)_{j \geq j_0}$ and $(\delta'_j)_{j \geq j_1}$ are equivalent. That is, there exists some $j' \geq j_1$ such that $p_{j'} \stateg p'_{j'}, \delta_{w_{j'}}(p'_{j'}) = p_{j'+1}$. Thus, by upward closure of $w_{j'}$, $p_{j'} \trans{w_{j'}:0} \ska[w_{j'}](p'_{j'}) = p_{j'+1}$. We can reiterate this process infinitely, taking $j_2 > j_1$, and so on, thus exhibiting an infinity of Büchi transitions in $(\delta_j)_{j \geq j_0}$.
		\end{proof}
		
		We obtain as a consequence of the two preceding lemmas that
	
	\detoB

	\knowledgenewrobustcmd\leqlexi[1][i]{\mathrel{\cmdkl{\leqslant_{\mathrm{lex}}^{#1}}}}
	\knowledgenewrobustcmd\llexi[1][i]{\mathrel{\kl[\leqlexi]{<_{\mathrm{lex}}^{#1}}}}
	\newrobustcmd\glexi[1][i]{\mathrel{\kl[\leqlexi]{>_{\mathrm{lex}}^{#1}}}}
	\newrobustcmd\geqlexi[1][i]{\mathrel{\kl[\leqlexi]{\leqslant_{\mathrm{lex}}^{#1}}}}
	
	As $\Ai$ recognises $\LangS(\A)$, "Eve-positional", by \Cref{theorem:co}, it is "$\eps$-completable". Actually, this "$\eps$-completion" is quite easy to describe and corresponds to the lexicographical order over "records" (seen as tuples).
	
	\begin{restatable}{lemma}{detepscomplet}
		Let $\Ai$ be the automaton obtained by the determinization of some "ordered Büchi automaton" $\A = (Q, \stateleq, I, \Gamma)$. Then $\Ai$ can be "$\eps$-completed" by adding the "$\eps$-transitions" in $\Delta_{\eps} := \{(s,\eps,2i,s') \mid s\glexi s'\}\ \cup\ \{(s',\eps,2i+1,s) \mid s\leqlexi s'\}$, where $\leqlexi$ corresponds to the lexicographical order over the first $i$ elements of a "record".
	\end{restatable}


\subsection{Optimality of the determinization for a complete enough alphabet}\label{sec:optim}

We observe that when the alphabet $\Gamma$ satisfies some completeness properties, this determinization procedure over an "ordered Büchi automaton" $\A = (Q, \stateleq, I, \Gamma)$ yields an automaton with a minimal number of states over all "deterministic parity automata" recognising $\LangS(\A)$. This is notably the case when $\Gamma$ contains all the tiles whose "skeleton" has all "horizontal transitions", all non-Büchi, except for possibly one (which is either Büchi or not present). In order to prove this minimality, we will use the flower game introduced by Colcombet and Zdanowski \cite{Safra_flowers}.

It is even possible to obtain this optimality over a larger subset of tile alphabets, by noticing that some alphabets cannot generate all possible permutations, and thus the determinization cannot reach all the "records". 

We first need to properly introduce games and positional strategies.

\subsubsection{Games and strategies}
For $\rho$ an infinite word and $n \in \NN$, $\rho_{|n}$ describes its finite prefix of length $n$.

\AP For $\Pi$ an "alphabet", $(V,E)$ a graph with $V$ a (potentially infinite) set of vertices and $\gamma:E\to \Pi \uplus \{\io\}$ an edge labeling, we call $G = (V,E,\gamma)$ a ""$\Pi$-graph"".

\AP For $L$ a language over $\Pi^\omega$, a ""$L$-game"" played by players Eve and Adam consists in a "$\Pi$-graph" $\G = (V_E \uplus V_A,E,\gamma)$ where the vertex set is partitioned in two subsets, $V_E$ and $V_A$, controlled respectively by Eve and Adam.
A ""play"" of $\G$ starting in $v\in V$ consists in an infinite sequence of edges $\rho := (e_i)_{i\in \NN}$ forming an infinite path starting in $v$. 
A "play" $(e_i)_{i\in \NN}$ is ""winning@@play"" for Eve (or simply \reintro*"winning@@play") if after some finite prefix all the $\gamma(e_i)$ are equal to $\io$, or if $(\gamma(e_i))_{i\in \NN} \in L$. Else it is said to be ""losing@@game"" (for Eve, and "winning@@game" for Adam).

\AP A ""strategy"" for Eve consists of a function $\sigma : E^* \to E$ such that, for all play $\rho$, for all $n \in \NN$, if $\rho_{|n}$ ends in a vertex $v \in V_E$, $\sigma(\rho_{|n})$ is an edge from $v$. A "play" $\rho$ is said to be ""consistent with"" the strategy $\sigma$ if for all $n$, $\rho_{|n}$ ending in a vertex of $V_E$ implies that $\rho_{|n+1} = \rho_{|n}\sigma(\rho_{|n})$. 
We say that an Eve "strategy" $\sigma$ is ""winning@winning strategy"" from vertex $v\in V$ if all plays "consistent with" $\sigma$ starting in $v$ are "winning@@play". We similarly define "strategies" for Adam, "winning@@strategy" when all plays "consistent with" them are "winning@@play" for Adam.
A "strategy" is ""positional@@strategy"" if it can be expressed as a function $\sigma : V \to E$, that is, if it does not depend on the path seen beforehand but simply on the current vertex.
For $L$ an "$\omega$-regular" language, all "$L$-games" are determined : from any vertex $v$, there exists one player with a winning strategy from $v$ \cite{Borel_determinacy}. The set of vertices such that Eve admits a winning strategy is ""Eve's winning region"", and we similarly define "Adam's winning region".


\subsubsection{Proof of optimality}

\knowledgenewcommand{\RA}{\cmdkl{R_\A}}
\knowledgenewcommand{\SR}{\cmdkl{S_R}}

In this section, we set some "ordered Büchi automaton" $\A = (Q,\stateleq, I, \Gamma)$.

\AP We denote by $\intro*\RA \subseteq Q$ its set of ""reachable residuals"". It is defined as the set of states $q$ such that, letting $q_I = \max_{\stateleq}(I)$, there exists $u \in \Gamma^*$ such that $q_I \trans{\skel(u)} q$. Intuitively, $q\in \RA$ if it is the maximum state of some reachable subset of $Q$. We can easily observe that for any $u\in \Gamma^*$ witnessing that some $q_u$ belongs to $\RA$, $q_u$ defines a "residual" for $L$, as $(Q, \stateleq, \{q' \mid q' \stateleq q_u\}, \Gamma)$ recognises the "residual" $u^{-1}L$. Note that $\RA$ depends on $I$.

\AP We define $\intro*\SR$ as the set of "records" $s$ over $Q$ such that $s(0) \in \RA$, along with the empty "record" if the semigroup generated by $(\Gamma, \cc)$ contains the empty tile.
We will show that the reachable states of $\Ai$ are in $\SR$, before showing that there is no "deterministic automaton" recognising $\LangS(\A)$ with less than $|\SR|$ states.

\begin{lemma}\label{lem:det-accessible-horizontal}
	\AP Let $\A = (Q,\stateleq, I, \Gamma)$ be an "ordered Büchi automaton" of "reachable residuals" $\RA \subseteq Q$. Then the automaton $\Ai$, obtained via determinization of $\A$, is such that all its accessible "records" in $\Ai$ are in $\SR$.
\end{lemma}
\begin{proof}
	We recall that the "initial state" $s_I$ of $\Ai$ has for first state $\max_{\stateleq}(I)$, which belongs to $\RA$ (using the "empty word" $\io$ as witness), therefore $s_I \in \SR$.
	
	We reason by induction on the length of "runs@@tile" in $\Ai$. As the "initial state" belongs to $\SR$, it suffices to show that all transitions from such "records" still lead to "records" in $\SR$. Let $w_j$ denote the word read so far.
	
	Let $s_j \in \SR$ be the current "record" and $t\in \Gamma$ be the tile read at time $j$, let us show that $s_{j+1} := \De(s_j,t) \in \SR$. We denote $P_j$ as in the determinization construction. 
	If $P_j$ is empty, then necessarily $s_{j+1}$ is the empty "record". In that case, the word $w_j t$ satisfies that there is no $q\in Q$ such that $\max_{\stateleq}(I) \trans{w_j t} q$. Therefore $w_j t \in \Gamma^+$ concatenates to the empty tile, and the empty record indeed belongs to $\SR$. Else, necessarily $s_{j+1}(0) = \ska[t] (s_{j}(0))$. But then, $w_j t$ acts as a witness that $s_{j+1}(0) \in \RA$, which concludes the proof.
\end{proof}
%

The idea behind the optimality proof is that, with an alphabet complete enough, for all "records" $s_1, s_2, s' \in S$ with $s_1 \neq s_2$ yet $s_1(0) = s_2(0)$, there exists a word $w\in \Gamma^+$ such that $\De(s_1, w) = \De(s_2,w) = s'$, yet the path from $s_1$ is dominated by an even parity and the one from $s_2$ is dominated by an odd priority. Therefore, intuitively, it is not possible for Eve to confuse both states without losing some expressivity.

A tile alphabet $\Gamma$ over a set of states $Q$ is said ""Rabin-complete"" if, for all $q\in Q$, it contains both the tile $\nu_q$ of "skeleton" $s_q := \{(q',1, q') \mid q' \neq q\}$, and the tile $\beta_q$ of "skeleton" $s_q \cup \{(q,0,q)\}$. Notably, up to adding a single state $\hat{q}$, maximal for ${\stateleq}$, such that for all tile $t$, $\hat{q} \trans{t:1} \hat{q}$, this alphabet can generate all the Rabin languages with $|Q|$ pairs.

\begin{lemma}\label{lem:det-horizon-optim}
	Let $\A = (Q,\stateleq, I, \Gamma)$ be an "ordered Büchi automaton" such that $\Gamma$ is "Rabin-complete". Then, $\LangS(\A)$ cannot be recognised by a "deterministic parity automaton" with fewer states than $|\SR|$.
\end{lemma}
\begin{proof}
	The proof scheme is an adaptation from Colcombet and Zdanowski's proof of the minimality of Safra's construction \cite{Safra_flowers}.
	
	In this proof, we denote $|Q|$ as $n$ and $|\SR|$ as $N$. We show that there does not exist a "deterministic automaton" $A'$ with $N'< N$ states, recognising $\LangS(\A)$. If it were the case, for $G$ a "$\overline{\LangS(\A)}$-game" won by Eve, she could win it with memory of size $N'$ (by taking the product of the complement of $A'$ with $G$, she obtains a parity game, whose positional winning strategy \cite{Emerson_Jutla_parity} can act as a winning strategy over $G$, of size $N'$). We reason by contradiction, and build an "$\overline{\LangS(\A)}$-game" that Eve cannot win with memory less than $N$.
	
	Before doing so, we require some preliminary observations. Let $s\in \SR$ be a "record" reachable in $\Ai$, and let $L_s \subseteq \Gamma^+$ be the language of finite non-empty tile sequences $w$ such that the run in $\Ai$ starting from $s$ along $w$ is dominated by some odd priority (recall that we are in min-parity), and ends in some $s'$ such that $s(0) = s'(0)$.
	As $L_s$ is a regular language, there exists a finite "game" $G_s$, with labels in $\Gamma$, such that paths along $G_s$ describe all the words in $L_s$.
	Note that $L_s$ cannot be the empty language as soon as $n > 1$\footnote{If $n=1$, $\Ai$ has a single state, and is immediately minimal}, because for $q \neq s(0)$, $\nu_q \in L_s$.
	
	Let $s = (q_0 q_1 \dots q_{k-1})$ be a "record" in $\SR$, for $i<k$ we denote $t_s^{(i)}$ the word whose "skeleton" contains exactly the "horizontal transitions" $q_0 \trans{1} q_0, \dots, q_i \trans{1} q_i$ (that is, up to considering the tile $t$ obtained by product of these tiles, $\ska[t]$ is the identity function over the subset $S_i := \{q_0, \dots, q_i\}$). We observe that by definition, $t_s^{(i)} \in \Gamma^*$, as we can simply obtain it as the concatenation of the $\nu_q$ for $q \notin S_i$.
	
	We can define the "$\overline{\LangS(\A)}$-game" $G$ in the following fashion. After an Adam-controlled stem of length at most $2n$, allowing Adam to choose any finite prefix of length $\leq 2n$, $G$ has a single vertex $v$ controlled by Eve, with $\io$-transitions leading towards $N$ other vertices $v_s$, each corresponding to some $s\in \SR$. From such a $v_s$, Adam can either go to $G_s$, or play any finite word $w$ such that $\img(\De(s,w)) \subsetneq \img(s)$. Whatever his choice, the play then ends up in $v$.
	A play is winning for Eve if the word described by this play is in $\overline{\LangS(\A)}$, as $G$ is a "$\overline{\LangS(\A)}$-game".
	
	\begin{figure}[!htb]
		\centering
		\includegraphics[width=0.6\textwidth]{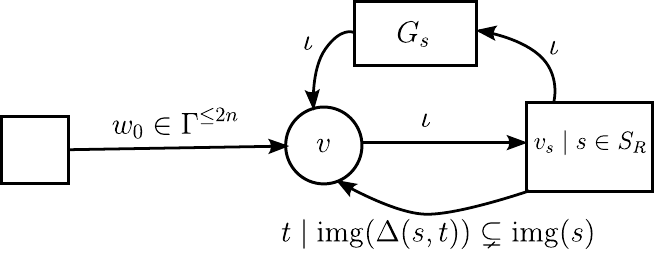}
		\label{fig:jeu-petales-horizontal}
		\caption{The "game" $G$, with one petal represented.}
	\end{figure}
	
	We now argue that this game is won by Eve: at any point during the play, when Eve has to make a choice (and is thus in $v$) after reading a prefix $w$, she goes to the vertex $v_s$, for $s := \De(s_I, w)$.
	Adam can only choose to decrease $\img(s)$ a finite number of times, as the only time $\img(s)$ can increase against that strategy for Eve is during the initial prefix (as the loop described by the $L_s$ preserves the image). Therefore, looking past the last time this image decreases, the resulting run is "parity rejecting" in $\Ai$, as a concatenation of segments each dominated by some odd priority. The corresponding word is therefore in $\overline{\LangS(\A)}$, and the play is indeed won by Eve.
	
	Let $\sigma$ be a "strategy" for Eve in $G$, of memory strictly smaller than $N$. Then, by a pigeonhole argument, there is some $s_0 = (q_0 \dots q_k)\in \SR$ such that no run compatible with $\sigma$ ever visits $v_{s_0}$. We denote $I_0 = \img(s_0)$.
	Let us build an Adam strategy winning against $\sigma$. The idea behind this strategy is that Adam will build a word whose run in $\Ai$ passes infinitely often via $s_0$. As Eve cannot choose $v_{s_0}$, she will choose at each step $j$ some different $v_{s_j}$. Adam will then either choose any word decreasing the image allowing to come back to $s_0$, if $\img(s_0)\subsetneq \img(s_j)$, or choose a word in the chosen $L_{s_j}$ that would be dominated by some even priority if seen from $s_0$, and end in $s_0$. The corresponding run in $\Ai$ is thus accepting, and the resulting word is in $\LangS(\A)$.
	
	Note that in $\Ai$, $s_0$ can be reached from $s_I$ whilst reading a word $w_0\in \Gamma^{\leq n+1}$: it suffices to define $w_0 = u_0 \nu_{q_1}\dots \nu_{q_{k-1}}$, where $u_0$ is a word witnessing the fact $q_0 \in \RA$: that is $q_I \trans{\skel(u_0)} q_0$, with $q_I := \max_{\stateleq} (I)$. Note that $u_0$ can always be chosen of length at most $n-1$.
	
	Let $v_s \neq v_{s_0}$ be a vertex chosen by Eve at some point during the play, after reading a prefix $w$ such that $\De(s_I, w) = I_0$. Let us prove that Adam can play a finite word $w'$ such that $\De(s_0, w') = s_0$ with an even corresponding priority.\\
	If $\img(s)\supsetneq I_0$, Adam can play the word $u_s := t_{s_0}^{(k-1)} \beta_{q_0}$. We observe that $\De(s_0, u_s) = s_0$, and thus $u_s$ strictly decreases the "record"'s image, as it ends up equal to $I_0$: $u_s$ is thus a legitimate move for Adam. Finally, as $s(0) \trans{u_s : 0} q_0 = s_0(0)$, the corresponding path in $\Ai$ has priority 0.\\
	Else, if $\img(s)\subsetneq I_0$, Adam can still play the word $u_s$. We easily observe that $u_s \in L_s$, and that $\De(s_0,t')=s_0$, once again with associated priority 0.\\
	Finally, if $\img(s) = I_0$, let us build $w_s \in L_{s}$ such that $\De(s_0, w_s) = s_0$ and the corresponding path is dominated by an even priority. Let $i \geq 1$ be the first index such that $s_0(i) \neq s(i)$ (recall that $\img(s) = \img(s_0)$ and thus $s_0(0) = s(0)$). We build the word $w_s= \beta_{q_i} t_{s_0}^{(i)} t_{s_0}^{(i+1)} \dots t_{s_0}^{(k-1)}$. We now argue that $w_s \in L_{s}$, and is thus a legitimate move for Adam. Indeed, the index $i$ is "forgotten", yet the first Büchi transition appears at index $s^{-1}(s_0(i)) > i$. Therefore, by minimality of $i$, the parity produced in $\Ai$ is thus $2i-1$, odd. Therefore, $w_s\in L_s$. The fact that both runs end in $s_0$ follows from the reset property of "records": since $q_0, \dots, q_{i-1}$ are preserved and all states of higher index are reintroduced in the correct order, the determinization reconstructs exactly $s_0$. Therefore, $\De(s_0, w_s) = s_0$, and this path is dominated by the priority $2i$: along this path no index $\leq i$ is "forgotten", and $i$ witnesses a Büchi transition.
	
	This defines a strategy for Adam that consists in playing $w_0$, then either $u_s$ or $w_s$ depending on the "record" $s$ currently chosen by Eve. This describes a cycle in $\Ai$ around $s_0$, dominated by an even priority. The resulting word is thus in $\LangS(\A)$. The strategy $\sigma$ therefore is not a winning strategy, contradiction, and Eve needs at least $N$ memory to win this game.
\end{proof}

From the two previous lemmas, we deduce that $\Ai$ necessarily has accessible state set $\SR$, and we thus obtain the minimality property.

\begin{lemma}\label{lem:det-optim}
	Let $\A = (Q,\stateleq, I, \Gamma)$ be an "ordered Büchi automaton" such that $\Gamma$ is "Rabin-complete". Then any "deterministic parity automaton" recognising $\LangS(\A)$ has at least as many states as $\Ai$.
\end{lemma}

From this result, we notably obtain that the procedure yields a parity automaton of minimal size for the $n$-complete Rabin, that is, the Rabin language with $n$ Rabin pairs, over the alphabet of size $3^n$ such that all the letters have a different behaviour.

	\appendix
	

\section{Proofs from \Cref{sec:Buchi tiles}}

\lemSiimonoid*
\begin{proof}
	By definition, $\ti \in \Sii$. 
	We first observe that $\ti = \{(p,c,q) \mid \exists q', (q',1,q') \transleq (p,c,q)\}$ is actually the set $\{(p,c,q) \mid q \stateleq p \text{and if } p = q \text{ then } c=1\}$. Indeed, for all $q \statel p$ and $c\in \{0,1\}$, $(p,1,p) \transleq (p,c,q)$, and there is no $q'$ such that $(q', 1, q') \transleq (q,c, p)$, nor such that $(q',1,q') \transleq (p,0,p)$.
	We then show that it is a neutral element for $\cc$ in $\Sii$.
	As $\unit \subseteq \ti$, we immediately observe that for all $t \in \Sii$, $t \subseteq \ti \cc t$ and $t \subseteq t \cc \ti$.
	Then, let $(p,c,q) \in \ti \cc t$. By definition, there exists $p' \in Q$ and $c_1, c_2 \in \inter{2}$ such that $(p,c_1, p') \in \ti$, $(p', c_2, q)\in t$ and $c = \min(c_1, c_2)$. By the previous observation, $p' \stateleq p$, and in case of equality, $c_1 = 1$. Therefore, $(p',c_2,q) \transleq (p,c,q)$, as $p' \stateleq p$, and in case of equality $c = c_2$ and thus $(p',c_2,q) = (p,c,q)$. By upward closure, we thus observe that $(p,c,q) \in t$. We similarly show that $t = t \cc \ti$.
	We finally prove that $\cc$ is internal to $\Sii$: for all $t,t'\in \Sii$, for all $\delta = (p,c,p') \in t \cc t'$ whose existence is witnessed by the transitions $(p,c_1, q)\in t$ and $(q, c_2, p') \in t'$, let us show that $t \cc t'$ contains all transitions $\delta^\dagger \transgeq \delta$. 
	Let $\delta^\dagger := (r,c_r, r')$ be such a transition. If $p \statel r$, then $\delta_1 := (r, c_r, q) \transgeq (p, c_1, q)$ and thus $\delta_1$ belongs to $t$ by upward closure. Similarly, as $r' \stateleq p'$, $\delta'_1 := (q, 1, r') \transgeq (q, c_2, p')$ and thus $\delta'_1$ belongs to $t'$. Therefore, $\delta^\dagger = (r, \min(c_r, 1), r') \in t \cc t'$, as witnessed by $\delta_1$ and $\delta'_1$. As $\min(c_r, 1) = c_r$, we thus obtain that $\delta^\dagger \in t \cc t'$. We reason symmetrically if $r' \statel p'$. Finally, if $r = p$ and $r' = p'$, then $c_r \geq c$, which concludes the proof by upward closure.
\end{proof}

\lemskel*
\begin{proof}
	Let $t$ be a "tile" in $\Sii$. If we project its "transitions" on a $Q \times Q$ plane (that is, a "transition" $(p,c,q)$ is projected to the point $(p,q)$), $t$ describes an upper-left-closed set, by upward closure of $t$ for ${\transleq}$. There thus exists a subset $S$ of incomparable points, of size at most $|Q|$, that can generate this whole set -- notably, they all have distinct first coordinate, and symmetrically for the second coordinate.
	If we recover the priority information, we observe that the set $s := \{(p,c,q) \mid (p,q)\in S \text{ and } \nexists c' < c, (p,c',q) \in t\}$ suffices to generate all the transitions in $t$, as we always consider the transition with lowest priority for each pair of states. By definition of $S$, $s\subseteq t$. We denote $\hat{t}$ as the upward closure of $s$, let us show that $t = \hat{t}$.
	
	We now show that $t \subseteq \hat{t}$. Indeed, for all $(p,c,q)\in t$, by definition of $S$, there is some $(p',q')\in S$ such that $p' \stateleq p$ and $q \stateleq q'$. Then, if either $p' \statel p$ or $q \statel q'$, as $(p', 1, q') \in \hat{t}$ (either as it directly belongs to $s$, or by upward closure), we get that $(p',1,q') \transleq (p,1,q)$ and thus $(p,1,q) \in \hat{t}$ by upwards-closure. Else, $p=p'$ and $q=q'$, in which case either $(p,0,q)\in s$ and thus $(p,c,q)\in \hat{t}$ by upward closure, or $(p,1,q)\in s$ and thus $c = 1$ by definition, which concludes as $s\subseteq \hat{t}$.\\
	Conversely, for $(p,c,q)\in \hat{t}$, if $(p,c,q) \in s$ then $(p,q)\in S$ and we easily observe that $c$ is such that $(p,c, q)\in t$. Else, there exists some $(p',c',q')\in s$ such that $(p',c',q')\transleq (p,c,q)$. By the previous argument, $(p',c',q')\in t$, and by upward closure, we thus obtain that $(p,c,q)\in t$.
	
	We finally observe the minimality of $s$ as a generator, as all the transitions of $s$ are incomparable: there is no strict subset of $s$ that still generates the missing transitions (which belong to $t$).
\end{proof}
	

\knowledgenewcommand\tA[1][A]{\cmdkl{\mathrm{t_{#1}}}}

\section{Proof of \Cref{thm:NPA-to-oB}}\label{sec:NPA-to-oB}

In this section we describe the transformation from "$\eps$-complete automata" to "ordered Büchi automata" described in \Cref{thm:NPA-to-oB}. It can be seen as a generalization of the construction described in \Cref{lem:Buchi-to-oB}. Once again, we will describe the behaviour of an "$\eps$-complete automaton" through the prism of an "ordered Büchi automaton". 
The main idea is that we collapse the parity condition into a binary condition, which forces us to encode priority information into the state space. Doing so requires more states to convey this expressive power and heavily relies on the tree structure defined by the "$\eps$-transitions" to describe the potential runs.

\AP Let $A = (Q_A, \Sigma, I_A, \Delta)$ be an "$\eps$-complete automaton" over some alphabet $\Pi$, of "index" $\inter{2k+2}$ for some $k\in \NN$. We assume without loss of generality that $I_A$ is downward-closed for the total preorder $\trans{\eps:2k+1}$.
We recall that for all $c$ odd in $\inter{2k+2}$, $\trans{\eps : c}$ defines a total preorder over $Q_A$, by "$\eps$-completeness". We denote the corresponding equivalence classes as ""$c$-classes"", or, more generally, as \reintro*"odd-classes".
We define its ""$\eps$-tree"" $\intro*\tA$, which makes explicit the stratification induced by its odd-priority "$\eps$-transitions". It corresponds to the ordered tree of depth $k$ that satisfies the following constraints:
\begin{asparaitem}
	\item It has, at depth $d>0$, a node $c_j$ for each "$2d-1$-class" $C_j$,
	\item each $c_j$ at depth $d>1$ is such that $C_j \subseteq C'$, with $C'$ the "odd-class" corresponding to its parent, and
	\item the children of a node at depth $d$ are ordered from left to right along the $\trans{\eps:2d}$ order.
\end{asparaitem}
The second condition can be satisfied as each $\trans{\eps:2d+1}$ refines $\trans{\eps:2d-1}$ (by definition of "$\eps$-complete automata"), and the third condition can be satisfied as $\trans{\eps:2d}$ corresponds to the strict variant of $\trans{\eps:2d+1}$ (this allows the children of a node to be linearly ordered).
An example of such a construction can be seen in \Cref{fig:automate-epsilon} and \Cref{fig:arbre-epsilon}.
\begin{figure}[!htb]
	\begin{minipage}[b]{0.55\textwidth}
		\includegraphics[width=\linewidth]{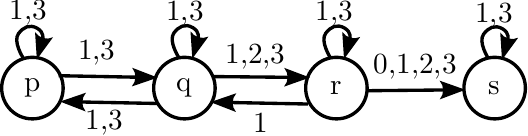}
		\caption{The $\eps$-transitions of some $\eps$-complete automaton (we do not represent the transitions that can be deduced by transitivity).}
		\label{fig:automate-epsilon}
	\end{minipage}
	\hfill
	\begin{minipage}[b]{0.35\textwidth}
		\includegraphics[width=\linewidth]{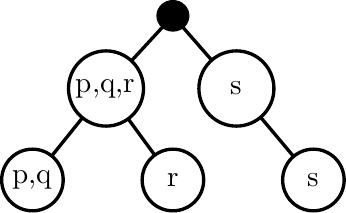}
		\caption{The corresponding $\eps$-tree. Its order ${\stateleq}$ is $\{p,q,r\}_1 \stateg \{p,q\}_2 \stateg \{r\}_2 \stateg \{s\}_1 \stateg \{s\}_2$, where the index corresponds to the node's depth.}
		\label{fig:arbre-epsilon}
	\end{minipage}
\end{figure}

Let us now define the desired "ordered Büchi automaton" $\A$, before establishing that its "language@@system" is $\Lang(A)$.

We consider the set $\Q$ of nodes in $\tA$ (root excepted). We will sometimes identify a node $Q$ at depth $d$ with its associated "$2d-1$-class". We define ${\stateleq}$ as the linear order over $\Q$ induced by the depth-first left-to-right traversal of $\tA$. Its greatest element is the highest $1$-class for $\trans{\eps:1}$ (from which all vertices are reachable by reading a $\trans{\eps:1}$). Once again, an example of this order can be found in \Cref{fig:arbre-epsilon}.
The ""initial set"" $I$ is defined as the downwards-closed set whose maximal element is the maximal $1$-class containing an element in $I_A$ (and, by downward closure of $I_A$, all the elements of classes in $I$ belong to $I_A$). Similarly as in \Cref{lem:Buchi-to-oB}, we will define below some function $f_t : \Pieps \to \Si$, and the morphism $f$ from $\Pieps$ to $\Sii$ is such that $f(a)$ is the upward closure of $f_t(a)$. Before defining $f_t$, we require one additional definition.

\knowledgenewrobustcmd\prioleq{\mathrel{\cmdkl{\unlhd}}}
\knowledgenewrobustcmd\priogeq{\mathrel{\kl[\prioleq]{\unrhd}}}
\knowledgenewrobustcmd\priolt{\mathrel{\kl[\prioleq]{\lhd}}}

\AP We recall what is often described as the ""preference order"" over a set of priorities. Given an index $\inter{2d+2}$, we define it as the total order $0 \intro*\priolt 2 \priolt 4 \dots 2d \priolt 2d+1 \priolt 2d-1 \dots 3 \priolt 1$. Intuitively, it corresponds to the order in which the existential player would rather see a priority in a game (or parity automaton).

For $a\in \Pieps$, for a depth $d\in[1,k]$, and for two nodes $Q,Q'\in \Q$ of depth $d$ in $\tA$, $Q \trans{f_t(a):1} Q'$ if there exists $q\in Q, q'\in Q'$ such that $q \trans{a:c} q'$ with $c\intro*\prioleq 2d-1$. If furthermore $c \prioleq 2d-2$, we also have $Q \trans{f_t(a):0} Q'$.
We then extend $f$ into a morphism over the free word monoid in the standard fashion.

The idea behind this transition is that a node $Q$ of depth $d$ will see a Büchi transition towards another node $Q'$ of depth $d$ over the tile $f_t(a)$ if and only if there is some path from the first class to the second one of even priority at most $2d-2$. Because both $Q$ and $Q'$ are $2d-1$-classes, a run can always move from an arbitrary $q_0\in Q$ to $q$, following the transition $q_0\trans{\eps:2d-1}q$. Then the $a$-transition of (even) priority $\prioleq 2d-2$ is taken, and finally another $\eps$-transition of priority $2d-1$ leads to the desired element of $Q'$.

\begin{remark}\label{rem:f-eps-is-ti}
	We observe that $f(\eps) = \ti$. Indeed, $f_t(\eps)$ contains all the transitions of the shape $Q \trans{f_t(\eps):1} Q$, and no strictly increasing transition over ${\stateleq}$ nor transition of the shape $Q \trans{f_t(\eps):0} Q$. Therefore $f(\eps) \in \Sii$ does not contain any transition $Q \trans{f(\eps):0} Q$ yet contains the upward closure of $\unit$: it is thus equal to $\ti$.
\end{remark}

In order to prove the language equivalence, we will show that runs in $A$ can be simulated in $\A$, and conversely. We first prove the two following lemmas, that give a better understanding of how runs in $A$ are mimicked in the "ordered Büchi automaton".
\begin{lemma}\label{lem:NPA-to-oB-access}
	Let $q,q'$ be two states in $Q_A$, and $u\in \Pieps^*$ be some finite word such that $q \trans{u} q'$. Let $Q, Q'$ be the "$1$-classes" such that $q\in Q, q'\in Q'$. Then $Q \trans{f(u)} Q'$.
\end{lemma}
\begin{proof}
	We prove this result by induction on the length of $u$.
	The base case is immediate.
	Else, for $u = u_0 a$ with $a \in \Pieps$, as $q \trans{u} q'$, there exists an intermediary state $q_0$ such that $q \trans{u_0} q_0 \trans{a} q'$. By induction hypothesis, for $Q_0$ the "$1$-class" containing $q_0$, $Q \trans{f(u_0)} Q_0$.
	Then, as $q_0 \trans{a:c} q'$ with $c \prioleq 1$ (as $1$ is maximal for $\prioleq$), $Q_0 \trans{f_t(a)} Q'$. 
	Therefore, as $f_t(a) \subseteq f(a)$, we still have $Q_0 \trans{f(a)} Q'$. Therefore $Q \trans{f(u_0)} Q_0 \trans{f(a)} Q'$, hence $Q \trans{f(u)} Q'$.
\end{proof}

We now prove that if a word admits a finite path dominated by an even priority in $A$, this property is mirrored in the "ordered Büchi automaton".
\begin{lemma}\label{lem:NPA-to-oB-accept}
	Let $q,q'$ be two states in $Q_A$, $2i\in \inter{2k+2}$ be some even "priority", and $u\in \Pieps^+$ be some finite word such that $q \trans{u:2i} q'$. Then, for $Q, Q'$ the "$2i+1$-classes" such that $q\in Q$ and $q'\in Q'$, $Q \trans{f(u):0} Q'$.
\end{lemma}
\begin{proof}
	We once more proceed by induction on the length of $u$.
	If $u$ is a single letter $a \in \Pieps$, then $q \trans{a:2i} q'$. As $2i \prioleq 2i$, $Q \trans{f_t(a):0} Q'$, and we similarly get that $Q \trans{f(a):0} Q'$.
	Else, $u$ is of the shape $u = u_0 a$ with $u_0\in \Pieps^+, a \in \Pieps$.
	As $u_0$ is non-empty and $q \trans{u:2i} q'$, there exists an intermediary state $q_0$ such that $q \trans{u_0:c} q_0 \trans{a:c'} q'$ with $c,c' \geq 2i$. Note that (at least) one of $c$ or $c'$ is exactly $2i$. By induction hypothesis, for $Q_0$ the $2i+1$-class containing $q_0$, $Q \trans{f(u_0)} Q_0$, and this path has a Büchi transition if $c = 2i$.
	Then, as $q_0 \trans{a:c'} q'$ with $c' \geq 2i$, we observe that $c' \prioleq 2i+1$ and thus $Q_0\trans{f_t(a)} Q'$, with this transition being Büchi if $c'=2i$. Finally, $f_t(a) \subseteq f(a)$, we still have $Q \trans{f(u_0)} Q_0 \trans{f(a)} Q'$. As at least one of $c$ or $c'$ is equal to $2i$, at least one of these paths is dominated by a Büchi transition, and thus $Q \trans{f(u):0} Q'$.
\end{proof}

Conversely, we also observe some simulation properties in the other direction, albeit less directly. Notably, we observe that the transitions taken along $\ti$ in the "ordered Büchi automaton" correspond to the $\eps$-transitions in $A$.

\begin{lemma}\label{lem:empty-oB-to-NPA}
	For $Q,Q' \in \Q$ two "odd-classes", if $Q' \stateleq Q$, then for all $q\in Q$ and $q'\in Q'$, $q \trans{\eps} q'$.
\end{lemma}
\begin{proof}
	We denote $2i+1$ and $2j+1$ the odd priorities of the classes $Q$ and $Q'$ respectively.
	Let $q\in Q, q'\in Q'$ be states in the corresponding classes. As $Q' \stateleq Q$, by definition of ${\stateleq}$ as the order in the left-to-right depth-first traversal of $\tA$, $Q'$ appears at the right of $Q$ or below it in $\tA$. In either case, by "$\eps$-completeness", for all $q\in Q$ and $q'\in Q'$, there exists an odd $\eps$ transition $q \trans{\eps:c} q'$ (if $Q'$ is below $Q$, it suffices to take $c=2i+1$, as $Q$ is an equivalence class for $\trans{\eps:2i+1}$). 
\end{proof}

We can now obtain an equivalent of \Cref{lem:NPA-to-oB-access} in the converse direction:
\begin{lemma}\label{lem:oB-to-NPA-access}
	Let $Q,Q' \in \Q$ be two "odd-classes". For all non-empty word $u = (u_i)_{i\leq N}\in\Pi^+$ such that $Q \trans{f(u)} Q'$, for all $q\in Q, q'\in Q'$, $q \trans{u'} q'$ where $u' := (\eps u_i \eps)_{i\leq N}$.
\end{lemma}
\begin{proof}
	We prove this result by induction on the length of $u$. The case $|u| = 1$ will be treated slightly differently.
	
	As $Q \trans{(f(u_i))_{i\leq N}} Q'$, there exists some class $Q_0$ such that $Q \trans{(f(u_i))_{i<N}} Q_0 \trans{f(a)} Q'$. As $Q_0 \trans{f(a):c} Q'$ with some priority $c$, by definition of $f$ there exists $\delta = (Q_1, c_t, Q_2) \in f_t(a)$ such that $\delta \transleq (Q_0, c, Q')$.
	Notably, $Q_1 \stateleq Q_0$ and $Q' \stateleq Q_0$.
	By definition of $f_t$, there exists some $q_1\in Q_1, q_2\in Q_2$ such that $q_1 \trans{a} q_2$.  Finally, by \Cref{lem:empty-oB-to-NPA}, for all $q_0\in Q_0$, $q_0 \trans{\eps} q_1$ and $\forall q'\in Q', q_2\trans{\eps} q'$. 
	If $u$ is of length one, the previous result suffices: we take $q_0$ to be equal to $q$, and thus obtain $q \trans{\eps} q_1 \trans{a} q_2 \trans{\eps} q'$.
	Else, $|u| > 1$, and by induction hypothesis, for all $q\in Q$, $q \trans{(\eps u_i \eps)_{i<N}} q_0$. Therefore $q \trans{(\eps u_i \eps)_{i<N}} q_0 \trans{\eps} q_1 \trans{a} q_2 \trans{\eps} q'$, which concludes the proof.
\end{proof}

We now establish, similarly, that the existence of a Büchi cycle in the "ordered Büchi automaton" implies the existence of an even cycle in $A$.
\begin{lemma}\label{lem:oB-to-NPA-even}
	Let $Q\in \Q$ be some "odd-class", and $u = (u_i)_{i\leq N}\in\Pi^+$ be a word such that $Q \trans{f(u):0} Q$. Then, for $d$ the minimal depth over "odd-classes" encountered along this path, for all $q\in Q$, there exists a cycle $q\trans{u' : c} q$, where $c \prioleq 2d-2$, and $u' := (\eps u_i \eps)_{i\leq N}$.
\end{lemma}
\begin{proof}
	According to \Cref{lem:oB-to-NPA-access}, there exists a path $\rho$ in $A$ from $q$ to $q$ over $u'$, composed of transitions on tiles $\eps$ and $(f_t(u_i))_{i\leq N}$. It remains to show that it can be built to have dominating priority $c \prioleq 2d-2$ (that is, $c$ is even and $c \leq 2d-2$).
	Let us first show that $\rho$ can be chosen such that it contains at least one transition of priority $\prioleq 2d-2$.
	
	We observe that for all $a \in \Pi$, by definition, $f_t(a)$ only has transitions remaining in "odd-classes" of the same level. Therefore, if there is a change in class level along $\rho$, it is along a transition generated by upward-closure.
	Let $(Q, f(u_j), Q')$ be a "Büchi transition" in $\rho$. If both $Q$ and $Q'$ are "$2d-1$-classes", then the corresponding transition in $A$ can be taken of priority $\prioleq 2d-2$. Indeed, either we directly have $Q \trans{f_t(u_j):0} Q'$, in which case by definition there exists $q, q'$ respectively in $Q$ and $Q'$ such that $q \trans{u_j:c'} q'$ with $c' \prioleq 2d-2$, or this Büchi transition belongs to $f(u_j)$ by upward closure witnessed by some transition $(Q_1, c_t, Q_2)\in f_t(u_j)$. In this case, we have two possibilities:  either $Q' \statel Q_2$, in which case the $\eps$-transition from $Q_2$ to $Q'$ is a transition between two different "$2d-1$-classes", and it can thus be chosen to be of priority $2d-2$; or $Q' = Q_2$ and $Q_1 \statel Q_0$, in which case $Q_2$ (and thus also $Q_1$) are "$2d-1$-classes", and the same argument applies for the transition between $Q_0$ and $Q_1$.\\
	Else, if there is no "Büchi transition" between classes of level $2d-1$, there still necessarily exists a Büchi transition in $\rho$, which thus takes place between classes of smaller level (for at least one of them) by minimality of $d$. As $\rho$ is a cycle, there is thus one transition starting in a class of level $> 2d+1$ and ending in a class of level $2d-1$ (by minimality of $d$). Then the same argument as above applies, and we once again can choose this transition to be of priority $2d-2$.
	
	Let us now show that the path in $A$ does not see an odd transition of priority $c'\leq 2d-3$ (that is, $c' \priogeq 2d-3$). Such a transition cannot occur during a transition over some $u_i$, as the path in the "ordered Büchi automaton" would fail to exist in the "tile" $f_t(u_i)$. Indeed, as $c' \priogeq 2d-3$, it would not have generated the corresponding transition in $f_t(u_i)$ (between classes of level at least $2d-1$), contradiction. 
	We can also choose $\rho$ such that no such transition occurs during one of the $\eps$-transitions : by construction, these transitions are always chosen either strictly decreasing along ${\stateleq}$ (in which case its priority can be chosen to be even), or as a transition from one class to itself. As all the classes considered are of level at least $2d-1$, by definition, any such transition can always be chosen of priority $\geq 2d-1$, which satisfies the desired property.
	
	Therefore, the corresponding path can be chosen to be dominated by an even priority of value at most $2d-2$.
\end{proof}

We can now establish that $\Lang(A)$ is "ordered Büchi".
\begin{lemma}
	Let $w\in \Pi^\omega$ be an infinite word, then $w\in \Lang(A)$ if and only if $f(w)\in \LangS(\A)$.
\end{lemma}
\begin{proof}
	As $\Lang(A)$ is "$\omega$-regular", we can restrict ourselves to the case where $w = (w_i)_{i\in \NN}$ is "ultimately-periodic", of the shape $u v^\omega$, as both $\Lang(A)$ and $\LangS(\A)$ are "$\omega$-regular". We denote $u'$ and $v'$ as the words obtained from $u$ and $v$ by replacing each letter $a$ with $\eps a \eps$.
	
	If $w\in \Lang(A)$, by \Cref{lem:eps-complete-intertwined}, the word $w' := (\eps, w_i, \eps)_{i\in \NN}$ admits an "accepting run@@aut" $\rho = \rho_u \rho_v^\omega \in \Delta_A^\omega$. Let us show that $f(w')\in \LangS(\A)$. This implies that $f(w)\in \LangS(\A)$: using \Cref{rem:f-eps-is-ti}, we can factor out all the $f(\eps) = \ti$, leaving only $f(w)$.
	We get, by \Cref{lem:NPA-to-oB-access}, that for $q$ the end vertex of $\rho_u$, of $1$-class $Q$, $I\trans{f(u)} Q$. Let $2d$ be the even priority dominating in $\rho_v$. Then, by \Cref{lem:NPA-to-oB-accept}, for $Q_d$ the $2d+1$-class containing $q$, $Q_d \trans{f(v'):0} Q_d$. It then suffices to observe that $f(w') = f(u')f(v')^\omega = f(u')\ti f(v')^\omega$, and by definition of the order ${\stateleq}$, as $Q$ and $Q_d$ are along the same branch of $\tA$ with $Q$ appearing first, $Q \trans{\ti} Q_d$. This exhibits an "accepting run@@tile" over $f(w')$, thus $f(w')\in \LangS(\A)$.

	Conversely, if $f(w)\in \LangS(\A)$, we show that $w'$ is recognised by $A$, where $w' := (\eps, w_i, \eps)_{i\in \NN}$. This suffices to show that $w$ is in the language of $A$ "$\eps$-automaton". As $f(w)\in \LangS(\A)$, there exists $Q\in \Q$ such that $I \trans{f(u) f(v)} Q$ and $Q \trans{f(v):0} Q$.
	We obtain by \Cref{lem:oB-to-NPA-access} that for all $q_i \in I_A$ and $q\in Q$, $q_{i} \trans{u' v'} q$ (note that we need to consider $u' v'$, as $u'$ may be the empty word). Then, by \Cref{lem:oB-to-NPA-even}, there exists some even $2d$ such that $q \trans{v' : 2d} q$. Therefore $w' = u' v'^\omega$ is recognised by $A$, which concludes the proof.
\end{proof}

Therefore, $\Lang(A)$ is indeed an "ordered Büchi language", and the corresponding "ordered Büchi automaton" has state set $\Q$, of size at most $kn$ -- being a tree with at most $n$ leaves and of depth $k$.

\section{Proofs from \Cref{sec:det}}

\AP In the following, for $i\in \inter{n}$, we denote $\intro*\leqlexi$ as the lexicographical order over the first $i+1$ elements of "records" in $S$ (if one of them has smaller length than the other, we can assume that its tail is filled with $\bot$). Its strict variant is denoted $\llexi$.

\detepscomplet
\begin{proof}
	For $\Ai = (S, \Sii, I, \De)$, we define $A' := (S, \Si, I, \De \uplus \Delta_{\eps})$, let us prove that it is an "$\eps$-completion" of $\Ai$.
	
	We immediately get the "$\eps$-completeness", as for all $s,s'\in S$ and even priority $2i$, either $s\trans{\eps:2i}s'$ or $s' \trans{\eps:2i+1} s$: Casares and Ohlmann proved that it is a necessary and sufficient condition for an automaton to be "$\eps$-complete" \cite[theorem 6.4]{CO24Positional}. Let us prove that adding these $\eps$-transitions does not change the recognised language.
	
	In this proof, $\forall s \in S$, we write $s[i]$ to describe the ordered tuple containing the $i$ first states of $s$.		
	We directly observe that, given some $i \in \inter{n}$, and looking at the equivalence classes described by $\leqlexi$ over the different $s[i]$ for $s \in S$, all the $\eps$-transitions of priority $2i$ are decreasing along $\leqlexi$, and those of priority $2i+1$ are non-increasing.
	
	Clearly, every "accepting run@@aut" of $\Ai$ is still an "accepting run@@aut" of $\A'$, hence $\Lang(\Ai) \subseteq \Lang(\A')$.
	
	Conversely, let $w\in \Lang(A')$ be a word "recognised@@aut" by $A'$. Therefore, there exists $w' \in ((\Sii)_{\eps})^\omega$ such that $w$ is obtained by removing all $\eps$'s from $w'$, and such that there exists $\rho'$ an "accepting run@@aut" of $\A'$ over $w'$, of dominating priority some even $2i$, passing by the vertices $(s'_j)_{j\in \NN}$. We define $\rho$ the (unique) run in $\Ai$ over $w$, passing by the vertices $(s_j)_{j\in \NN}$.
	
	By definition of $w'$, there exists an increasing function $g:\NN \to \NN$ such that $\forall j \in \NN, w_j = w'_{g(j)}$ and $\forall j \in \NN, w'_j \in \Sii \implies j \in \img(g)$.
	
	We observe that all the $\eps$-transitions $s'_j \trans{\eps} s'_{j+1}$ in $\Delta_{\eps}$ are such that $s'_j(0) \stategeq s'_{j+1}(0)$ -- or, equivalently, $\img(s'_j) \supseteq \img(s'_{j+1})$, as by definition $s(0) = \max(\img(s))$.
	Reasoning by induction, by monotonicity of the function $w_j : \cP(Q) \to \cP(Q)$, we can observe that for all $j$, $\img(s_j) \supseteq \img(s'_{g(j)})$.
	
	As eventually no priority smaller than $2i$ is witnessed in $\rho'$, past some point $k'$ (that we can choose as the image $g(k)$ of some $k\in \NN$, as $g$ diverges), the prefixes $(s'_j[i])_{j \geq k'}$ are equal. We denote $w' = u'v'$, where $u'$ is the prefix of $w'$ of length $k'$, and similarly $w = uv$ for $u$ the prefix of length $k$.
	
	We observe that past $k'$, all the $\eps$-transitions in $\rho'$ are of priority at least $2i+1$ and thus do not affect the first $i$ indices. We also note that as $2i$ is witnessed infinitely often, therefore the "track" defined by $(s'(i)_j)_{j\geq k'}$ is infinite and witnesses infinitely many Büchi transitions.
	As the $\eps$-transitions after time $k$ do not affect these first $i$ indices, we observe that this result still holds for $v$, defined as $w$ with all its $\eps$'s removed. Therefore, $v$ is accepting in $\Ai_{s'_{k'}} := (S, \Sii, \{s'_{k'}\}, \De)$, defined as $\Ai$ of initial state $s'_{k'}$. By \Cref{lem:Ai-correct}, we thus get that for $\A'_{k'}:= (Q, \stateleq, \img(s'_{k'}), \Gamma)$, $v \in \LangS(\A'_{k'})$. Finally, as by the previous induction $\img(s_k) \stategeq \img(s'_{k'})$, and as the language recognised by an "ordered Büchi automaton" is monotone in its "initial states" $I$ (by \Cref{rem:langs-inclusion}), we get that $v\in \LangS((Q, \stateleq, \img(s_k), \Gamma))$, which concludes the proof as $s_k = \De(s_I, u)$.
\end{proof}
\end{document}